\newcommand{\longversion}[1]{#1}
\newcommand{\shortversion}[1]{}

\documentclass[a4paper,UKenglish,cleveref, autoref, thm-restate]{lipics-v2021}

\hideLIPIcs  

\usepackage{graphicx} 
\usepackage{todonotes}

\definecolor{thechosenone}{rgb}{0,0.1,0.5}
\title{Minimizing Envy and Maximizing Happiness in Graphical House Allocation} 

\author{Anubhav Dhar}{Indian Institute of Technology Kharagpur, India}{anubhavldhar@gmail.com}{https://orcid.org/0009-0006-5922-8300}{}

\author{Ashlesha Hota}{Indian Institute of Technology Kharagpur, India}{ashleshahota@gmail.com}{https://orcid.org/0009-0009-8805-4583}{}

\author{Palash Dey}{Indian Institute of Technology Kharagpur, India \and \url{https://cse.iitkgp.ac.in/~palash/} }{palash.dey@cse.iitkgp.ac.in}{https://orcid.org/0000-0003-0071-9464}{}

\author{Sudeshna Kolay}{Indian Institute of Technology Kharagpur, India \and \url{https://cse.iitkgp.ac.in/~skolay/} }{skolay@cse.iitkgp.ac.in}{https://orcid.org/0000-0002-2975-4856}{}

{
}{
}{}

\authorrunning{P. Dey, A. Dhar, A. Hota, S. Kolay} 



\category{} 

\relatedversion{} 

\nolinenumbers 

\EventEditors{}
\EventNoEds{2}
\EventLongTitle{}
\EventShortTitle{}
\EventAcronym{}
\EventYear{}
\EventDate{}
\EventLocation{}
\EventLogo{}
\SeriesVolume{}
\ArticleNo{}

\usepackage{longtable}

\definecolor{thechosenone}{rgb}{0,0.1,0.5}

\newcommand{\CLIQUE}{{{CLIQUE}}}
\newcommand{\ohaab}{{\sc Annnotated Optimal House Allocation of Agent Network}\xspace}

\usepackage{longtable}
\usepackage{hyperref}
\usepackage{amssymb,amsmath}
\usepackage{color}
\usepackage{makecell}
\usepackage{mathtools}
\usepackage{bbm}
\usepackage{setspace}

\usepackage{tikz}
\usetikzlibrary{arrows,positioning}

\newcommand{\defproblem}[3]{
  \vspace{1mm}
\begin{center}
\noindent\fbox{

  \begin{minipage}{0.95\textwidth}
  \begin{tabular*}{\textwidth}{@{\extracolsep{\fill}}l} \textsc{\underline{#1}} \\ \end{tabular*}\vspace{1ex}
  {\bf{Input:}} #2  \\
  {\bf{Question:}} #3
  \end{minipage}

  }
\end{center}
  \vspace{1mm}
}

\newdimen\prevdp
\def\leftlabel#1{\noalign{\prevdp=\prevdepth
   \kern-\prevdp\nointerlineskip\vbox to0pt{\vss\hbox{\ensuremath{#1}}}\kern\prevdp}}

\usepackage{url}

\usepackage{enumerate}

\usepackage{xspace}
\newcommand{\NP}{\ensuremath{\mathsf{NP}}\xspace}

\newcommand{\NPH}{\ensuremath{\mathsf{NP}}\text{-hard}\xspace}

\newcommand{\tw}{\text{tw}\xspace}

\let\oldlambda\lambda
\renewcommand{\lambda}{\ensuremath{\oldlambda}\xspace}
\let\oldalpha\alpha
\renewcommand{\alpha}{\ensuremath{\oldalpha}\xspace}
\let\oldDelta\Delta
\renewcommand{\Delta}{\ensuremath{\oldDelta}\xspace}

\newcommand{\ohaa}{{\sc Optimal House Allocation of Agent Network}\xspace}
\newcommand{\ohaah}{{\sc Optimally Happy House Allocation of Agent Network}\xspace}
\newcommand{\ohha}{{\sc Optimal House Network Allocation}\xspace}

\newcommand{\ehaa}{{\sc Egalitarian House Allocation of Agent Network}\xspace}
\newcommand{\ehha}{{\sc Egalitarian House Network Allocation}\xspace}

\newcommand{\uhaa}{{\sc Utilitarian House Allocation of Agent Network}\xspace}
\newcommand{\uhha}{{\sc Utilitarian House Network Allocation}\xspace}

\renewcommand{\AA}{\ensuremath{\mathcal A}\xspace}

\newcommand{\EE}{\ensuremath{\mathcal E}\xspace}
\newcommand{\FF}{\ensuremath{\mathcal F}\xspace}
\newcommand{\GG}{\ensuremath{\mathcal G}\xspace}
\newcommand{\HH}{\ensuremath{\mathcal H}\xspace}

\newcommand{\OO}{\ensuremath{\mathcal O}\xspace}
\newcommand{\PP}{\ensuremath{\mathcal P}\xspace}

\usepackage{nicefrac}

\newcommand{\eps}{\ensuremath{\varepsilon}\xspace}
\renewcommand{\epsilon}{\eps}

\newcommand{\ignore}[1]{}

\renewcommand{\ge}{\geqslant}
\renewcommand{\le}{\leqslant}

\usepackage{enumitem}
\setlist[enumerate]{labelwidth=!, labelindent=0pt}

\usepackage{rotating}
\usepackage{amssymb}
\usepackage{latexsym}
\usepackage{epsfig}
\usepackage{xcolor}
\usepackage{url}
\usepackage{caption}
\usepackage{subcaption}
\usepackage{array}
\usepackage{multirow}
\usepackage{enumerate}
\usepackage{pdflscape}
\usepackage{amsmath,graphicx,algorithm,algpseudocode,amsfonts}
\usepackage{epstopdf}
\usepackage{float}
\floatstyle{ruled}
\newfloat{Algorithms}{!thb}{lop}
\floatname{Algorithms}{Algorithm}

\allowdisplaybreaks[1]

\algnewcommand\algorithmicinput{\textbf{Input:}}
\algnewcommand\INPUT{\item[\algorithmicinput]}
\algnewcommand\algorithmicoutput{\textbf{Output:}}
\algnewcommand\OUTPUT{\item[\algorithmicoutput]}
\algnewcommand{\LineComment}[1]{\State \(\triangleright\) #1}

\usepackage{pifont}

\usepackage{mathtools}

\usepackage{algorithm}
\usepackage{algpseudocode}


\usepackage{comment}

\usepackage{cleveref}
\definecolor{cblue}{RGB}{0, 0, 128}
\crefname{theorem}{Theorem}{\bf Theorems}
\crefname{observation}{Observation}{\bf Observations}
\crefname{corollary}{Corollary}{\bf Corollary}
\crefname{lemma}{Lemma}{\bf Lemmata}
\crefname{corollary}{Corollary}{\bf Corollaries}
\crefname{proposition}{Proposition}{\bf Propositions}
\crefname{definition}{Definition}{\bf Definitions}
\crefname{claim}{Claim}{\bf Claims}
\crefname{reductionrule}{Reduction rule}{\bf Reduction rules}
\usepackage{color,soul} 

\begin{document}

\maketitle

\begin{abstract}
In this paper, we study a generalization of the {\sc House Allocation} problem. In our problem, agents are represented by vertices of a graph $\GG_{\mathcal{A}} = (\AA, E_\AA)$, and each agent $a \in \AA$ is associated with a set of preferred houses $\PP_a \subseteq \HH$, where $\AA$ is the set of agents and $\HH$ is the set of houses. A house allocation is an injective function $\phi: \AA \rightarrow \HH$, and an agent $a$ envies a neighbour $a' \in N_{\GG_\AA}(a)$ under $\phi$ if $\phi(a) \notin \PP_a$ and $\phi(a') \in \PP_a$. We study two natural objectives: the first problem called \ohaa, aims to compute an allocation that minimizes the number of envious agents; the second problem called \ohaah aims to maximize, among all minimum-envy allocations, the number of agents who are assigned a house they prefer. These two objectives capture complementary notions of fairness and individual satisfaction.

Our main results provide a comprehensive algorithmic and complexity-theoretic understanding of these objectives. First, we design polynomial time algorithms for both problems for the variant when each agent prefers exactly one house. On the other hand, when the list of preferred houses for each agent has size at most $2$ then we show that both problems are \NP-hard even when the agent graph $\GG_\AA$ is a complete bipartite graph. We also show that both problems are \NP-hard even when the number $|\mathcal H|$ of houses is equal to the number $|\mathcal A|$ of agents. This is in contrast to the classical {\sc House Allocation} problem, where the problem is polynomial time solvable when $|\mathcal H| = |\mathcal A|$. The two problems are also \NP-hard when the agent graph has a small vertex cover. On the positive side, we design exact algorithms that exploit certain structural properties of $\GG_{\AA}$ such as sparsity, existence of balanced separators or existence of small-sized vertex covers, and perform better than the naive brute-force algorithm. 
\end{abstract}

\newpage
\section{Introduction}\label{sec:intro}

The challenge of allocating limited resources in a fair and efficient manner arises in a wide variety of real-world contexts, from assigning dorm rooms to students, to distributing public housing, to allocating virtual machines in cloud systems. A particularly well-studied problem in this area is the {\sc House Allocation} problem, in which we are given a set $\mathcal A$ of agents and a set $\mathcal H$ of houses. An allocation of houses to agents is simply an injective function from agents to houses. We are also given for each agent $a$ a list $\mathcal{P}_a$ of preferred houses. The agent is satisfied if they are allocated a house from $\mathcal{P}_a$. On the other hand, if agent $a$ is not allocated a house from $\mathcal{P}_a$, then they will envy any other agent $b$ who gets allocated a house from $\mathcal{P}_a$. The objective is the find an allocation of houses to agents that minimizes the number of envious agents.

Ideally, we would like to obtain an allocation of houses to agents such that no agent is envious. Unfortunately, envy-free allocations are not guaranteed to exist, even when the number of agents and the number of houses are equal. 

Notice that the House Allocation problem assumes that agents have global information about all other agents, which may be impractical in the real world. It is more likely that agents have information about the allocation of other agents that they know of, like their friends or co-workers. Thus, in this paper we study a generalization of the {\sc House Allocation} problem, called the \ohaa (also referred to as the {\sc Graphical House Allocation} problem in literature). In this problem, we have an agent graph $\GG_{\mathcal{A}}=(\mathcal A, E_{\mathcal A})$. In our variant, an agent $a$ is only envious of an agent $b$ if all three of the following conditions are satisfied: (i) $a$ is not allocated a house from $\mathcal{P}_a$, (ii) $b$ is allocated a house from $\mathcal{P}_a$ and (iii) $b$ is a neighbour of $a$ in the graph $\GG_{\mathcal A}$. Now, the question is to find an allocation that minimizes the number of envious agents. Note that when the agent graph $\GG_{\mathcal{A}}$ is a complete graph then \ohaa is exactly the {\sc House Allocation} problem.

There is also a notion of happy agents in allocation problems. An agent $a$ is happy when they are allocated a house from $\mathcal{P}_a$. Note that an agent $a$ may both be non-envious of any other agent as well as unhappy. This could happen if no house in $\mathcal{P}_a$ is allocated to any agent in the closed neighbourhood of $a$ in the graph $\GG_{\mathcal{A}}$. While reduction of envy is important, it is also important to maximize the number of agents that are happy. A concrete application arises in university dormitory assignments, where students submit room preferences and are socially connected through shared courses or clubs. In such settings, it is crucial to limit perceived unfairness among friends while ensuring that as many students as possible are assigned to desirable rooms. This dual criterion captures realistic concerns of both equity and satisfaction in constrained allocation environments. Therefore, our second problem, \ohaah asks for an allocation that, among all minimum envy allocations, has the maximum number of happy agents. 

It is easy to see that any NP-hardness result of \ohaa also translates to NP-hardness of \ohaah. Also, any algorithmic result of \ohaah translates to an algorithm for \ohaa. Now, we describe our results.

\subsection{Contributions}

First, we consider both \ohaa and \ohaah when each agent prefers exactly one house. In other words, for each agent $a$, $|\mathcal{P}_a| = 1$. We design polynomial time algorithms for both problems under this variant. The details of this are discussed in Section~\ref{sec:asp-graph}, as Theorem~\ref{thm:d-is-1}.

Next, we show that even when the preference lists of each agent is at most $2$, the problem becomes \NPH. It is to be noted that Madathil et. al~\cite{madathil2024cost} proved NP-hardness for {\sc House Allocation} (\ohaa on agent graphs which are complete graphs) when the preference list of each agent is at most $2$. Our NP-hardness result for \ohaa holds even when the agent graph is a $3$-regular graph. The same NP-hardness result shows that \ohaa, and consequently \ohaah, is \NPH even when the number $|\mathcal{A}| = n$ of agents is equal to the number $|\mathcal{H}| = m$ of houses. This is in stark contrast to the {\sc House Allocation} problem, which can be solved in polynomial time when $n = m$. These proofs appear in detail as Theorem~\ref{thm:bip-d-2} and Theorem~\ref{thm:3reg-hard}, in Section~\ref{sec:asp-graph}.

Thus, we turn to exact algorithms for \ohaa and \ohaah. Note that there is a trivial brute force algorithm that runs in time $2^{\OO(n \log m)}$, where $n$ is the number of agents and $m$ is the number of houses. First, we design an algorithm for \ohaa that on input $\GG_{\mathcal{A}} = (\mathcal A, E_{\mathcal A})$ and $\mathcal H$ solves the problem in time $2^{|\mathcal{A}|+2|E_{\mathcal{A}}|}\cdot (|\mathcal{A}|+|\mathcal H|)^{\OO(1)}$. This algorithm performs better than the naive algorithm for all graphs when $|E_{\mathcal{A}}|=o(|\mathcal A| \log |\mathcal H|)$. The details of this algorithm can be found in Theorem~\ref{thm:exp-edge}, Section~\ref{sec:exact}.

Next, we look at other structural properties of the agent graph and design exact algorithms that utilize the structure. We design an alternate algorithm that has the best running time when the agent graph comes from a class of graphs that have small balanced separators. For example, this algorithm is the most efficient when the agent graph is a planar graph or a bounded genus graph. The details of this algorithm can be found in Theorem~\ref{thm:sep-rec}, Section~\ref{sec:exact}.

We further design an exact exponential algorithm for \ohaa for agent graphs with bounded vertex cover size. If the size of the vertex cover of the agent graph is $(\log n) ^{\OO(1)}$ then this algorithm is the fastest algorithm. On the other hand, we also show that \ohaa is \NPH even when the agent graph is a split graph or a complete bipartite graph with maximum vertex cover size at most $n^{\varepsilon}$, where $n$ is the number of agents and $\varepsilon \in (0,1)$ is a constant - therefore, we do not expect to obtain quasi-polynomial time algorithms for the problem when the vertex cover size of the agent graph is $n^{\varepsilon}$.


\subsection{Related Work}

The {\sc House Allocation} problem and its various extensions have been extensively studied across economics and computer science~\cite{abraham2004pareto,beynier2019local,dai2024weighted,gerards2025envy,hosseini2023tight,kamiyama2021complexity}, particularly in contexts that emphasize fairness and efficiency. We discuss the most relevant prior works to our setting.

Abdulkadiro\u{g}lu and S\"onmez~\cite{abdulkadirouglu1999house} introduce the \emph{House Allocation with Existing Tenants} model, which captures scenarios where some agents already occupy houses and may choose to keep them or participate in reallocation. They propose the \emph{Top Trading Cycles (TTC)} mechanism, which guarantees allocations that are Pareto efficient, individually rational, and strategy-proof. Their model generalizes the classical house allocation setting by incorporating endowments and offering mechanisms that incentivize participation without loss.
Focusing on the classical {\sc House Allocation} problem under ordinal preferences, Gan, Suksompong, and Voudouris~\cite{gan2019envy} develop a polynomial-time algorithm to determine whether an envy-free allocation exists when each agent ranks houses. They explore conditions under which envy-free allocations can or cannot be guaranteed and highlight cases where approximations or partial allocations become necessary. Their work contributes foundational results on the computational aspects of envy-freeness under structured preference inputs and motivates the need for relaxed fairness notions studied in later works.

Hosseini et al.~\cite{hosseini2023graphical,hosseini2024graphical} introduce the {\sc Graphical House Allocation} model, wherein agents are vertices in a graph, and envy is permitted only along the graph’s edges. Assuming identical valuation functions across agents, they focus on minimizing the \emph{total envy}, measured as the sum of envy across edges. They show that the problem generalizes the classical linear arrangement problem and establish hardness results even for simple graphs like disjoint paths and stars. To address tractability, they develop fixed-parameter algorithms based on graph properties like separability and provide polynomial-time solutions for special families such as paths and cycles. Choo, Ling, Suksompong, Teh, and Zhang~\cite{choo2024envy} address the challenge of achieving envy-freeness in settings where such allocations may not naturally exist. They propose the use of \emph{subsidies}---monetary compensations that accompany house allocations---to compute envy-free outcomes while minimizing the total subsidy required. They show that the problem is NP-hard in general but solvable in polynomial time when the number of houses differs from the number of agents by a constant, or when agents have identical utilities.

Madathil, Misra, and Sethia.~\cite{madathil2024cost} provide a comprehensive study on minimizing envy in {\sc House Allocation}, introducing three formal measures of envy---number of envious agents, maximum envy, and total envy. They define and analyze the corresponding computational problems (Optimal House Allocation, Egalitarian House Allocation, and Utilitarian House Allocation), establish strong hardness results even under restricted settings like binary preferences and bounded degrees, and design polynomial-time algorithms for structured instances such as extremal preferences. They also present fixed-parameter tractable algorithms based on agent and house types, formulate ILP models with empirical evaluations, and study the trade-off between fairness and welfare through the Price of Fairness, offering tight bounds and identifying cases where fairness incurs no welfare loss.

\section{Preliminaries}\label{sec:prelims}

Let \(\AA\) be a set of \(n\) agents and \(\HH\) be a set of \(m\) houses. Each agent \(a \in \AA\) is associated with a set of \emph{preferred houses} \(\PP_a \subseteq \HH\). Agent $a$ prefers houses in \(\PP_a\) while disliking the rest. A \emph{house allocation} is a function \(\phi: \AA \rightarrow \HH\) that is injective, meaning no two agents are assigned the same house. In this paper, $|\mathcal A|$ and $|\mathcal{H}|$ shall also be denoted by $n$ and $m$, respectively. 

We denote by \(\GG_\AA = (\AA, E_\AA)\) a graph over the agent set \(\AA\), modeling the social or informational structure among agents. An agent \(\AA\) is said to \emph{envy} a neighbour \(a' \in N_{\GG_\AA}(a)\) under an allocation \(\phi\) if \(\phi(a) \notin \PP_a\) and \(\phi(a') \in \PP_a\); that is, \(\AA\) does not receive a house they like, while their neighbour does. Let \(\EE_\AA^\phi(a)\) denote the set of agents in \(N_{\GG_\AA}(a)\) whom agent \(\AA\) envies under allocation \(\phi\), i.e.,

\(
\EE_\AA^\phi(a) = \{ a' \in N_{\GG_\AA}(a) \mid \phi(a) \notin \PP_a \text{ and } \phi(a') \in \PP_a \}.
\)

We call an agent $a \in \AA$ \emph{envious} if $\left|\EE_\AA^\phi(a)\right| \ge 1$, otherwise, $a$ is said to be \emph{non-envious} or \emph{envy-free}. We call an agent $a \in \AA$ \emph{happy}, if $\phi(a) \in \PP_a$. The \emph{envy} (resp. \emph{happiness}) of the allocation $\phi$ is the total number of envious (resp. happy) agents for the allocation $\phi$. A house is called a \emph{dummy} house, if it is not preferred by any agent; otherwise we call it a \emph{non-dummy} house.

We will use the notation $[t]$ to denote the set $\{1, 2, \ldots, t\}$. For any function $f: X \to Y$, for a subset $X' \subseteq X$, we denote by $f(X')$ the set of images of $X'$, i.e. $f(X') = \{f(x) \mid x \in X'\}$.

We are interested in computing allocations \(\phi\) that minimize various measures of envy in this graphical setting. Below, we define two optimization problems under this natural perspective capturing envy across agent relationships (defined by \(\GG_\AA\)). For completeness, we define each of these problems as follows:

\defproblem{\ohaa}{
A set \(\AA\) of \(n\) agents, a graph \(\GG_\AA\) on \(\AA\), a set \(\HH\) of \(m\) houses, and agent preferences \((\PP_a)_{a \in \AA}\).}{
Compute an allocation \(\phi\) that minimizes \(\left| \{ a \in \AA \mid \EE_\AA^\phi(a) \ne \emptyset \} \right|\).
}

\defproblem{\ohaah}{
A set \(\AA\) of \(n\) agents, a graph \(\GG_\AA\) on \(\AA\), a set \(\HH\) of \(m\) houses, and agent preferences \((\PP_a)_{a \in \AA}\).}{
Compute an allocation \(\phi\) that minimizes \(\left| \{ a \in \AA \mid \EE_\AA^\phi(a) \ne \emptyset \} \right|\), and among all such minimum-envy allocations, maximizes \(\left| \{ a \in \AA \mid \phi(a) \in \PP_a \} \right|\).
}

Note that \ohaah is at least as hard as \ohaa, as any solution to \ohaah is also a solution to \ohaa. Whenever we provide algorithms, we analyze this for both \ohaa and \ohaah. However, for hardness results, we will only focus on \ohaa, as these results would automatically translate to \ohaah as well.

Unless otherwise stated, we assume that preferences are represented extensionally by the preferred set \(\PP_a\) for each agent \(\AA\), and that all agents receive exactly one, distinct, house. We also assume that envy is only locally defined over \(\GG_\AA\), capturing realistic scenarios where agents can compare only within their social or spatial neighbourhoods. Throughout the paper, we use the following additional notation. For any agent \(a \in \AA\), we write \(\deg(a)\) to denote the degree of \(\AA\) in the graph \(\GG_\AA\). We denote by $d = \max \limits_{a \in \AA} |\PP_a|$, the maximum number of houses that a single agent prefers. 
An allocation is said to be \emph{zero-envy} if no agent envies any of their neighbours, i.e., \(\EE_A^\phi(a) = \emptyset\) for all \(a \in \AA\). 
Finally, for any house \(h \in \HH\), we use \(\phi^{-1}(h)\) to denote the agent assigned to house \(\HH\) under \(\phi\), if any.

\section{Introducing an Agent Graph: Contrast with the Classical Setting}\label{sec:asp-graph}

Notice that the House Allocation problem studied by Madathil et. al~\cite{madathil2024cost} is a special case of \ohaa where the agent graph is a complete graph. In this section, we provide results for \ohaa that depict that this has a much more interesting complexity landscape in comparison to the {\sc House Allocation} problem.

We first look into a tractable variant of {\sc House Allocation} that remains tractable for \ohaa . Madathil et. al~\cite{madathil2024cost} show that when each agent prefers exactly one house, \ohaa on complete graphs is polynomial-time solvable. We start by extending this result to solving \ohaa on arbitrary graphs in polynomial time, when each agent prefers exactly one house i.e. $d=1$.

\begin{theorem}\label{thm:d-is-1}
    There is a polynomial-time algorithm for \ohaa,  when every agent prefers exactly one house ($d=1$).
\end{theorem}

\begin{proof}

    Consider the instance $(\AA, \HH, \GG(\AA, E_\AA), (\PP_a)_{a \in \AA})$ of \ohaa. Since each agent $a \in \AA$ likes exactly one house, $a$ can be envious of at most one other agent. That is, $|\EE_\AA^\phi(a)| \le 1$. We will crucially use this observation to construct a polynomial-time algorithm.

    Allocating house $h \in \HH$ to agent $a \in \AA$ causes the set of agents $Y_{a,h} = \{a' \in N_{\GG_\AA}(a) \mid \PP_{a'} = \{h\}\}$ to be envious of $a$. Note that in such a case, the agents in $Y_{a,h}$ are not envious of any other agents. This observation would be crucially used to avoid double counting of envious agents. 


    We now construct a weighted complete bipartite graph $\GG_\FF$ with vertex partitions $\AA$ and $\HH$. The weights $w(a, h)$ for $a \in \AA$, $h \in \HH$ are defined as $w(a, h) = |Y_{a,h}|$. 
    The motivation behind defining such a weight function is to model through $w(a, h)$, the number of agents that become envious on assigning house $h \in \HH$ to agent $a \in \AA$. We finally use a polynomial-time minimum cost maximum bipartite matching algorithm to get a minimum-envy assignment. We formally prove this as follows.

    \begin{claim}\label{clm:d-is-1-matching}
        A minimum-envy allocation has the number of envious agents equal to the cost of the minimum cost maximum matching of $\GG_\FF$. The matching itself corresponds to such an allocation. 
    \end{claim}
    \begin{proof}[Proof of Claim~\ref{clm:d-is-1-matching}]
        It suffices to show that every maximum matching in $\GG_\FF$ has a bijective relation with a allocation of houses to agents, where the cost of the matching is equal to the number of envious agents in the allocation. 

        Let $M$ be a maximum matching; since $|\AA| \le |\HH|$, each agent must be matched to exactly one house. Define the allocation $\phi$ which assigns each agent $a$, its matched house in $\HH$. The weight of any matched edge $\{a, h\}$ denotes the number of envious agents due to assigning $h$ to $a$. Moreover, since each agent can be envious of at most one other agent, the total number of envious agents is simply the sum of the weights of the matching edges: the cost of the matching $M$.

        On the other hand, if $\phi$ is an allocation, then consider a matching $M = \{\{a, \phi(a)\} \mid a \in \AA\}$ defined by this allocation. This is a matching as no two agents are allocated the same house. If an agent $a'$ is envious in $\phi$, it is envious of exactly one of its neighbours, $a \in N_{\GG_\AA}(a')$. This happens when $\PP_{a'} = \{h\}$ and $\phi(a) = h$. Therefore, $a' \in Y_{a,h}$ and hence this contributes exactly $1$ to the weight $w(a, \phi(a))$ while not affecting any other matching edge. Indeed, every such envious agent contributes exactly $1$ to exactly one of the matching edges. This again implies that the total number of envious agents in $\phi$ is same as the cost of the matching $M$.
        
        This completes the proof of Claim~\ref{clm:d-is-1-matching}.
    \end{proof}
    Thus we are done.
\end{proof}
    For completeness, we provide the pseudocode of this polynomial-time algorithm, called Algorithm~\ref{alg:d-is-1}.

    \begin{algorithm}[htbp]
        \caption{ \hfill \textbf{Input:} $\AA, \HH, \GG_\AA(\AA,E_\AA), (\PP_a)_{a \in \AA}$ \hfill \textbf{Output:} Minimum envy}\label{alg:d-is-1}
        \begin{algorithmic}[1]
            \For{$a \in \AA$, $h \in \HH$}
                \State{$w(a, h) \gets |\{a' \in N_{\GG_\AA}(a) \mid \PP_{a'} = \{h\}\}|$}
            \EndFor
            \State{$\GG_\FF \gets$ complete bipartite graph with partitions $\AA$, $\HH$, and weights $w$}
            \State{\Return Minimum cost maximum matching of $\GG_\FF$}
        \end{algorithmic}
    \end{algorithm}

Note that Algorithm~\ref{alg:d-is-1} can be modified to solve \ohaah just by altering the definition of $w(a,h)$ to:

\[ 
    w(a, h) = \begin{cases} 
    |\{a' \in N_{\GG_\AA}(a) \mid \PP_{a'} = \{h\}\}| - \frac{1}{n+1} & \text{if } h \in \PP_a\\
    |\{a' \in N_{\GG_\AA}(a) \mid \PP_{a'} = \{h\}\}| & \text{if } h \notin \PP_a
\end{cases}
\]

A matching in such a bipartite graph will have a cost of $\left(\alpha - \beta \cdot \frac{1}{n + 1}\right)$, where $\alpha, \beta \in [n]$ if it corresponds to an allocation with envy $\alpha$ and happiness $\beta$.

\begin{corollary}
    There is a polynomial-time algorithm for \ohaah,  when every agent likes exactly one house ($d=1$).
\end{corollary}

Now we study \ohaa when the input graph is restricted to certain simple graph classes. For instance, in the case of independent sets, \ohaa is trivially solvable as all allocations are zero-envy, and \ohaah is simply equivalent to the maximum cardinality bipartite matching problem. Moreover, restricting to some arbitrary graph class might potentially increase or decrease the hardness of the problem relative to when the input graph is a complete graph. Hence it is interesting to explore the complexity of \ohaa when the input graph is restricted to certain simpler graph classes.

Madathil et. al.~\cite{madathil2024cost} show that \ohaa for cliques is hard even when each agent prefers at most two houses, i.e. $d=2$. We show that even for the class of bipartite graphs, \ohaa remains hard even when each agent prefers at most two houses, i.e. $d = 2$. This involves a reduction from the \NPH problem of \CLIQUE, which is a nontrivial extension of the work by Madathil et. al.~\cite{madathil2024cost}. 

\shortversion{
\begin{theorem}[$\star$]\label{thm:bip-d-2}
    Given an agent graph $\GG_\AA = (\AA, E_\AA)$ and a set $\HH$ of houses, \ohaa is \NP-hard, even when $\GG_\AA$ is a complete bipartite graph where each agent prefers at most two houses.
\end{theorem}

\begin{proof}[Proof Sketch]
    We provide a reduction from the \NP-hard problem of \CLIQUE~on regular graphs to \ohaa. Let $(G(V,E), k)$ be an instance of \CLIQUE~on regular graphs where $G$ is $\delta$-regular and the problem asks to decide if $G$ has a clique of size $k$. Let $|V| = N$ and $|E| = M$. We reduce $(G(V,E), k)$ to an instance of \ohaa as follows. 

    \begin{itemize}
        \item Let $\AA = \{a^j_v \mid v \in V, j \in [\delta]\} \cup \{a_e \mid e \in E\}$. That is, we have an agent $a_e$ for every edge $e$ of $G$, and $\delta$ agents $a^1_v, a^2_v, \ldots, a^\delta_v$ for each vertex $v$ of $G$. Therefore, $n = |\AA| = \delta N + M$. 
        \item Let $\HH = \{h_v \mid v \in V\} \cup \{h^*_j \mid j \in [\delta N+M-k]\}$. Therefore, $|\HH| = N + \delta N + M - k = |\AA| + (N-k)$. 
        \item We are now going to set the houses $h_v$ to be preferred by agents $a^j_v$ for all $j \in [\delta]$, and also by $a_e$, for all edges $e$ which are incident on $v$. This is effectively done by setting $\PP_{a^j_v} = \{h_v\}$, and $\PP_{a_e} = \{h_u, h_v\}$, where $e = \{u,v\}$. The houses $\{h^*_j \mid j \in [\delta N + M - k]\}$ are preferred by no agents, i.e. these are dummy houses.
        \item We now define the underlying graph $\GG_\AA(\AA, E_\AA)$. Let $\AA_V = \{a^j_v \mid v \in V, j \in [\delta]\}$ and let $\AA_E = \{a_e \mid e \in E\}$. Then $\GG_\AA$ defines the complete bipartite graph between $\AA_V$ and $\AA_E$. That is, $E_\AA = \{\{a,a'\} \mid a \in \AA_V, a' \in \AA_E\}$. 
        \item The target number of envious agents is $k\delta - \binom{k}{2}$. Not that for non-trivial instances, $k \le \delta$, hence $k\delta - \binom{k}{2}$ is a positive integer.
    \end{itemize}
    
    Please refer to Figure~\ref{fig:bip-d-2} in the Appendix for clarity. Every step can be done in polynomial time; hence this is a polynomial-time reduction. The full proof of correctness of the reduction can be found in the Appendix.
\end{proof}}
\begin{theorem}\label{thm:bip-d-2}
    Given an agent graph $\GG_\AA = (\AA, E_\AA)$ and a set $\HH$ of houses, \ohaa is \NP-hard, even when $\GG_\AA$ is a complete bipartite graph where each agent prefers at most two houses.
\end{theorem}
\begin{figure}[ht]
    \begin{subfigure}[t]{0.95\textwidth} 
        \centering 
        \includegraphics[width=0.3\textwidth]{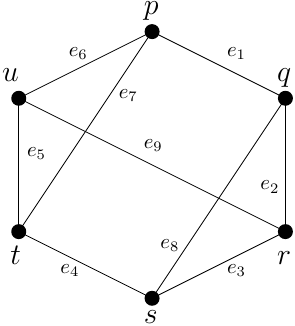}
        \caption{Instance of CLIQUE for $3$-regular graphs, $k = 3$}
    \end{subfigure} \\
    \begin{subfigure}[t]{0.95\textwidth}
        \centering
        \includegraphics[width=0.9\textwidth]{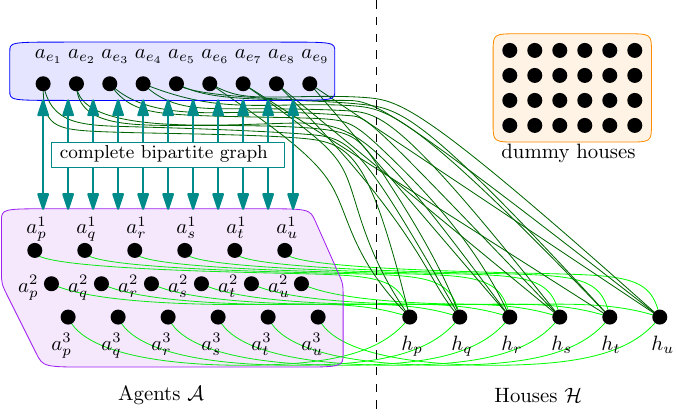}
        \caption{The reduced instance as per Theorem~\ref{thm:bip-d-2}}
    \end{subfigure} \hfill \hfill \hfill
    \caption{Reduction for Theorem~\ref{thm:bip-d-2}} \label{fig:bip-d-2}
\end{figure}

\begin{proof}
    We provide a reduction from the \NP-hard problem of \CLIQUE~on regular graphs to \ohaa. Let $(G(V,E), k)$ be an instance of \CLIQUE~on regular graphs where $G$ is $\delta$-regular and the problem asks to decide if $G$ has a clique of size $k$. Let $|V| = N$ and $|E| = M$. 

    We reduce $(G(V,E), k)$ to an instance of \ohaa as follows. 

    \begin{itemize}
        \item Let $\AA = \{a^j_v \mid v \in V, j \in [\delta]\} \cup \{a_e \mid e \in E\}$. That is, we have an agent $a_e$ for every edge $e$ of $G$, and $\delta$ agents $a^1_v, a^2_v, \ldots, a^\delta_v$ for each vertex $v$ of $G$. Therefore, $n = |\AA| = \delta N + M$. 
        \item Let $\HH = \{h_v \mid v \in V\} \cup \{h^*_j \mid j \in [\delta N+M-k]\}$. Therefore, $|\HH| = N + \delta N + M - k = |\AA| + (N-k)$. 
        \item We are now going to set the houses $h_v$ to be preferred by agents $a^j_v$ for all $j \in [\delta]$, and also by $a_e$, for all edges $e$ which are incident on $v$. This is effectively done by setting $\PP_{a^j_v} = \{h_v\}$, and $\PP_{a_e} = \{h_u, h_v\}$, where $e = \{u,v\}$. The houses $\{h^*_j \mid j \in [\delta N + M - k]\}$ are preferred by no agents, i.e., these are dummy houses.
        \item We now define the underlying graph $\GG_\AA(\AA, E_\AA)$. Let $\AA_V = \{a^j_v \mid v \in V, j \in [\delta]\}$ and let $\AA_E = \{a_e \mid e \in E\}$. Then $\GG_\AA$ defines the complete bipartite graph between $\AA_V$ and $\AA_E$. That is, $E_\AA = \{\{a,a'\} \mid a \in \AA_V, a' \in \AA_E\}$. 
        \item The target number of envious agents is $k\delta - \binom{k}{2}$. Not that for non-trivial instances, $k \le \delta$, hence $k\delta - \binom{k}{2}$ is a positive integer.
    \end{itemize}
    
    Please refer to Figure~\ref{fig:bip-d-2} for clarity. Every step can be done in polynomial time; hence this is a polynomial-time reduction. We now show the correctness of the reduction. 

    \textbf{Forward direction.}
    Let $S \subseteq V$ be a clique in $G$ of size $|S| = k$. Consider the allocation $\phi$ defined as follows:
    \begin{itemize}
        \item Assign $h_v$ to $a^1_v$ for all $v \in S$. 
        \item For the remaining $|\AA| - k = \delta N+M-k$ agents, assign them a dummy house each. 
    \end{itemize}
    
    In $\phi$ the houses $h_v$ are left unassigned for $v \notin S$. No agent in $\AA_V$ are envious as all their neighbours (i.e. $\AA_E$) are assigned dummy houses. Moreover, $a_e$ is not envious if $e$ is not incident on $S$; this is because the houses preferred by $a_e$ are unallocated. Therefore, the number of envious agents is at most the number of edges incident on $S$; this is precisely $k\delta - \binom{k}{2}$, as $G$ is $\delta$-regular and $S$ is a clique of size $k$.

    \textbf{Reverse direction.}
    Let $\phi$ be an allocation with at most $k\delta - \binom{k}{2}$ envious agents. We say that an allocation $\phi$ is nice if for all $v \in V$, the house $h_v$ is either unassigned, or assigned to the agent $a^1_v$. If $\phi$ is not nice to begin with, then $\phi$ can be converted to a nice allocation by a sequence of exchanges, none of which increase total envy.

    Let $h_v$ be a house that is assigned to an agent in $\AA$ other than $a^1_v$.

    \begin{itemize}
        \item \textbf{Case I: $h_v$ is assigned to $a^j_v$ for some $j \ne 1$.} We swap the houses allocated to $a^j_v$ and $a^1_v$. This does not change the number of envious agents, due to the symmetry of $\GG_\AA$.
        \item \textbf{Case II: $h_v$ is assigned to $a^j_u$ for some $j \in [\delta]$, $u \ne v$.} We swap the houses allocated to $a^j_u$ and $a^1_v$. Again, this does not increase the number of envious agents. This is because
        \begin{enumerate}
            \item $a^1_v$ is not envious in both allocations.
            \item $a^j_u$ did not have their preferred house before the swap. Since the neighbourhood of $a^j_u$ is unaffected by the swap, $a^j_u$ cannot become envious after the swap if $a^j_u$ was not envious before the swap.
            \item For all other agents, they are envious after the swap if and only if they were envious before the swap. This is because the set of houses allocated to their neighbours remain unchanged.
        \end{enumerate}
        \item \textbf{Case III: $h_v$ is assigned to $a_e$ for some $e \in E$.} $a^1_v, a^2_v \ldots, a^\delta_v$ must be envious of $a_e$ in $\phi$. We now look at the sequence of agents $a^{1}_{v_0} = a^1_v, a^{1}_{v_1}, \ldots, a^{1}_{v_{p-1}}, a^{1}_{v_p}$ such that 
        \begin{itemize}
            \item $\phi(a_e) = h_{v_0}$.
            \item $\phi(a^{1}_{v_0}) = h_{v_1}, \phi(a^{1}_{v_1}) = h_{v_2}, \ldots, \phi(a^{1}_{v_{p-1}}) = h_{v_p}$.
            \item $\phi(a^{1}_{v_p}) = h^*_q$ is a dummy house.
        \end{itemize}

        Note that such a sequence must exist and end with an agent who is assigned a dummy house, as the number of agents and houses are finite. Therefore, $p$ is a finite integer. Now, we modify the allocation to $\phi'$ as follows:

        \begin{itemize}
            \item $\phi'(a_e) = h^*_q$.
            \item $\phi'(a^{1}_{v_0}) = h_{v_0}, \phi'(a^{1}_{v_1}) = h_{v_1}, \ldots, \phi'(a^{1}_{v_{p-1}}) = h_{v_{p-1}}, \phi'(a^{1}_{v_p}) = h_{v_p}$.
            \item $\phi'(a) = \phi(a)$ for all other agents $a$.
        \end{itemize}

        We will now argue why this set of exchanges do not increase the number of envious agents.
        \begin{itemize}
            \item \textsl{Case A: $e$ is incident on $v$.} Let $e = e_1, e_2, \ldots, e_\delta$ be the edges incident on $v$. In $\phi$, $a_e$ was not envious, but agents $a_{e_2}, \ldots, a_{e_\delta}$ may or may not be envious. The agents $a^1_v, a^2_v, \ldots, a^\delta_v$ were all envious in $\phi$. In $\phi'$, $a^1_v, a^2_v, \ldots, a^\delta_v$ are not envious, agent $a_e$ becomes envious, and agents $a_{e_2}, \ldots, a_{e_\delta}$ may become envious. Therefore, the number of envious agents in $\{a_e, a_{e_2}, a_{e_3}, \ldots, a_{e_\delta}\} \cup \{a^1_v, a^2_v, \ldots, a^\delta_v\}$ do not increase. The agents $a^1_{v_1}, a^1_{v_2}, \ldots, a^1_{v_p}$ all become non-envious in $\phi'$ as they get their preferred houses. For all other agents $a^j_u$ where $u$ is a vertex in $V$, the dummy house $h^*_q$ is the only house which appears in the set of houses assigned to neighbours of $a^j_u$ in $\phi'$ which did not appear in $\phi$; hence $a^j_u$ is envious in $\phi'$  only if they were envious in $\phi$. Finally, all other agents $a_{e'}$ where $e'$ is an edge in $E$, are envious in $\phi'$ if and only if they were envious in $\phi$ as only $h_v$ (a house not preferred by $a_{e'}$) gets added to the set of houses assigned to their neighbours. 
            \item \textsl{Case B: $e$ is not incident on $v$.} Let $e_1, e_2, \ldots, e_\delta$ be the edges incident on $v$. In $\phi$, the agents $a_{e_1}, a_{e_2}, \ldots, a_{e_\delta}$ may or may not be envious; but the agents $a^1_v, a^2_v, \ldots, a^\delta_v$ were all envious in $\phi$. In $\phi'$, $a^1_v, a^2_v, \ldots, a^\delta_v$ are not envious, while agents $a_{e_1}, a_{e_2}, \ldots, a_{e_\delta}$ may become envious. Therefore, the number of envious agents in $\{a_{e_1}, a_{e_2}, a_{e_3}, \ldots, a_{e_\delta}\} \cup \{a^1_v, a^2_v, \ldots, a^\delta_v\}$ do not increase. The agents $a^1_{v_1}, a^1_{v_2}, \ldots, a^1_{v_p}$ all become non-envious in $\phi'$ as they get their preferred houses. If the agent $a_e$ is not envious in $\phi$, then it cannot become envious in $\phi'$ as only $h_{v}$ (a house not preferred by $a_e$) gets added to the set of houses assigned to its neighbours. For all other agents $a^j_u$ where $u$ is a vertex in $V$, the dummy house $h^*_q$ is the only house which appears in the set of houses assigned to neighbours of $a^j_u$ in $\phi'$ which did not appear in $\phi$; hence $a^j_u$ is envious in $\phi'$ only if they were envious in $\phi$. Finally, all other agents $a_{e'}$ where $e'$ is an edge in $E$, are envious in $\phi'$ if and only if they were envious in $\phi$ as only $h_v$ (a house not preferred by $a_{e'}$) gets added to the set of houses assigned to their neighbours. 
        \end{itemize}
    \end{itemize}

    These exchanges do not increase the total envy, but they allocate at least one house $h_v$ to $a^1_v$, which was previously allocated to some other agent. Repeating this for at most $|\HH|$ times would give us a nice allocation, with envy no more than that in $\phi$.

    Hence, we can safely assume that $\phi$ is a nice allocation with envy at most $k\delta - \binom{k}{2}$. $\phi$ allocates dummy houses to $a_e$ for all $e \in E$ and to $a^2_v, a^3_v, \ldots, a^\delta_v$ for all $v \in V$. Moreover, it is safe to assume that all dummy houses are assigned to some agent, otherwise we can assign a dummy house to any $a^1_v$ without increasing total envy.

    We now look into agents which are not assigned dummy vertices by $\phi$. There are $|\AA| - (\delta N + M - k) = (\delta N + M) - (\delta N + M -k) = k$ many such agents, all of the form $a^1_v$, $v \in V$. Let $a^{1}_{v_1}, a^{1}_{v_2}, \ldots, a^{1}_{v_k}$ be the agents which are not assigned dummy houses. For any $e \in E$, $e$ is incident on least one vertex in $v_1, v_2, \ldots, v_k$, if and only if the agent $a_e$ is envious. Therefore, the number of envious agents is precisely the number of edges which are incident on the vertices $v_1, v_2, \ldots, v_k$. Let $\gamma$ be the number of edges with both endpoints in $\{v_1, v_2, \ldots, v_k\}$. The number of edges incident on $\{v_1, v_2, \ldots, v_k\}$ is therefore equal to $k\delta - \gamma$. By assumption, the number of envious agents is at most $k\delta - \binom{k}{2}$. Thus $\gamma \ge \binom{k}{2}$, implying $v_1, v_2, \ldots, v_k$ is a clique.

    This completes the proof of correctness of the reduction.
\end{proof}

In Section~\ref{sec:vc}, we further prove that restricting the agent graphs to bipartite graphs (Theorem~\ref{thm:vc-bip}) and split graphs (Theorem~\ref{thm:vc-split}) with very small vertex cover sizes still yields \NP-hardness.

Moreover, Madathil et. al~\cite{madathil2024cost} provide polynomial-time algorithms for the following cases: (i) when $|\AA| = |\HH|$, (ii) when preferences $(\PP_a)_{a \in \AA}$ follow an extremal structure. Preferences follow an extremal structure if there is an ordering of houses in $\HH = \{h_1, h_2, \ldots, h_m\}$ such that for every agent $a \in \AA$, they prefer either some suffix of the houses, $\PP_a = \{h_i, h_{i+1}, \ldots, h_m\}$, or some prefix of the houses $\PP_a = \{h_1, h_2, \ldots, h_i\}$, for some $i \in [m]$. 

In contrast, we show that even when we restrict the graph of an \ohaa instance to be 3-regular, and have $|\AA| = |\HH|$ and identical preferences for agents, the problem is \NPH. Note that identical preferences, i.e. $\PP_a = \PP_{a'}$ for all $a, a' \in \AA$, is a special case of extremal structure as we can order the houses where all dummy houses are ordered after all houses preferred by every agent. 

\shortversion{
\begin{theorem}[$\star$]\label{thm:3reg-hard}
    Given an agent graph $\GG_\AA = (\AA, \EE_\AA)$ and a set $\HH$ of houses, \ohaa is \NP-hard even when the input graph is a 3-regular graph with $|\AA| = |\HH|$ and agents have identical preferences (i.e. $\PP_a = \PP_{a'}$ for all $a, a' \in \AA$).
\end{theorem}}

\longversion{
\begin{theorem}\label{thm:3reg-hard}
    Given an agent graph $\GG_\AA = (\AA, \EE_\AA)$ and a set $\HH$ of houses, \ohaa is \NP-hard even when the input graph is a 3-regular graph with $|\AA| = |\HH|$ and agents have identical preferences (i.e. $\PP_a = \PP_{a'}$ for all $a, a' \in \AA$).
\end{theorem} 
\begin{proof}
    We provide a reduction from the $1/2$-Vertex Separator problem~\cite{muller1991alpha} on 3-regular graphs to  \ohaa. The problem takes input a graph $G(V,E)$ and an integer $k$ and asks if there is a subset $S$ of $V$ that $|S| \le k$ such that $V \setminus S$ can be partitioned into two equal-sized subsets $S_1$, $S_2$ such that no edge connects a vertex in $S_1$ with a vertex in $S_2$. This problem is shown to be \NP-complete even for 3-regular graphs by M{\"u}ller and Wagner~\cite{muller1991alpha}.

    Let $(G(V,E),k)$ be an instance of $1/2$-Vertex Separator problem such that $G(V,E)$ is 3-regular. Let $V = \{a_1, a_2, \ldots, a_n\}$. We create an instance $(\AA, \HH, \GG_\AA(\AA, E_\AA), (\PP_a)_{a \in \AA}, 2\left\lfloor \frac k 2 \right\rfloor)$ of \ohaa as follows:

    \begin{itemize}
        \item $\AA = V = \{a_1, a_2, \ldots, a_n\}$.
        \item $\HH = \{h_1, h_2, \ldots, h_n\}$.
        \item $E_\AA = E$.
        \item Let $\HH^* = \{h_1, h_2, \ldots, h_t\}$, where $t = \frac{n}{2} - \left\lfloor \frac{k}{2} \right\rfloor$. $\PP_a = \HH^*$ for all $a \in \AA$. 
    \end{itemize}

    Intuitively the underlying agent graph is the same as $G$, and every agent prefers the same $t = \frac{n}{2} - \left\lfloor \frac{k}{2} \right\rfloor$ houses. Note that $t$ is an integer since the number of vertices of $G$, $n$, is even as $G$ is 3-regular. This reduction can be computed in polynomial time. We now show the correctness of the reduction. 
    
    \begin{claim}\label{clm:half-vs-equiv}
        $(G, k)$ and $(G, 2 \left\lfloor \frac{k}{2} \right\rfloor)$ are equivalent instances of the $1/2$-Vertex Separator problem, if $G$ is 3-regular.
    \end{claim}
    \begin{proof}
        This holds trivially when $k$ is even. Further, if $(G, 2 \left\lfloor \frac k 2 \right\rfloor)$ is a Yes-instance, then so is $(G, k)$ as $k \ge 2 \left\lfloor \frac k 2 \right\rfloor$.
        
        Now, say $k$ is odd. and let $(G, k)$ be a Yes-instance. Let $k = 2\lambda + 1$ for $\lambda \in \mathbb{Z}$. Note that since $G$ is 3-regular, $n = |V|$ must be even. Let $S$ be any separator that separates $V\setminus S$ into equal sized, disjoint, and non-adjacent parts $S_1$, $S_2$ and satisfying $|S| \le k$. Therefore, $|S_1| = |S_2|$. Hence $|V| - |S| = |S_1| + |S_2|$, or $|S| = n - 2|S_1|$, an even integer: say $|S| = 2 \beta$, for some integer $\beta$. Therefore, $2\beta \le 2\lambda + 1 \implies \beta \le \lambda + 0.5$. Since $\beta, \lambda$ are integers, we must have $\beta \le \lambda$; implying $|S| = 2\beta \le 2\lambda = 2 \left\lfloor \frac k 2 \right\rfloor$. Hence $(G, 2 \left\lfloor \frac k 2 \right\rfloor)$ is also a Yes-instance.  
    \end{proof} 

    \begin{claim}\label{clm:half-exact}
        If $(G(V,E), 2 \left\lfloor \frac{k}{2} \right\rfloor)$ is a Yes-instance of the $1/2$-Vertex Separator problem, where $G$ is 3-regular, there exists $S \subseteq V$ of size $|S| = 2 \left\lfloor \frac k 2 \right\rfloor$ that separates $V \setminus S$ into equal sized, disjoint, and non-adjacent parts.
    \end{claim}
    \begin{proof} 
        Let $S'$ be any separator that separates $V \setminus S$ into two equal sized, disjoint, and non-adjacent parts $S_1', S_2'$, such that $|S'| \le 2 \left\lfloor \frac k 2 \right\rfloor$. Such an $S'$ exists as $(G(V,E), 2 \left\lfloor \frac k 2 \right\rfloor)$ is a Yes-instance. 

        Let $S_1' = \{x_1, x_2, \ldots, x_s\}$ and $S_2' = \{y_1, y_2, \ldots, y_s\}$ for some $s$, where $x_i$ is not adjacent to $y_j$ for all $i,j \in [s]$. Therefore, $|S'| = n - 2s$, an even number. This gives us

        $$n - 2s \le 2 \left\lfloor \frac k 2 \right\rfloor \implies s - \frac n 2 + \left\lfloor \frac k 2 \right\rfloor \ge 0 $$

        Define $\gamma = s - \frac n 2 + \left\lfloor \frac k 2 \right\rfloor$, note that $0 \le \gamma \le s$. We now construct our target separator $S$.

        \begin{itemize}
            \item define $S_1 = \{x_{\gamma + 1}, x_{\gamma + 2}, \ldots, x_s\}$
            \item define $S_2 = \{y_{\gamma + 1}, y_{\gamma + 2}, \ldots, y_s\}$
            \item define $S = S' \cup \{x_1, x_2, \ldots, x_{\gamma}\} \cup \{y_1, y_2, \ldots, y_\gamma\}$
        \end{itemize}
         
        We have $S \cup S_1 \cup S_2 = V$, and $S$ separators $V \setminus S$ into equal sized, disjoint, and non-adjacent subsets $S_1$ and $S_2$. Moreover, 
        
        $$|S| = |S'| + 2\gamma = (n - 2s) + 2 \left(s - \frac n 2 + \left\lfloor \frac k 2 \right\rfloor\right) = 2 \left\lfloor \frac k 2 \right\rfloor$$

        This completes the proof of Claim~\ref{clm:half-exact}.
    \end{proof}

    \textbf{Reverse direction.}
    We first show if the \ohaa instance $(\AA, \HH, \GG_\AA(\AA, E_\AA), (\PP_a)_{a \in \AA}, 2 \left\lfloor \frac k 2 \right\rfloor)$ is a Yes-instance, then so is the $1/2$-Vertex Separator instance of $(G,k)$. Let $\phi$ be an allocation with at most $2\left\lfloor \frac k 2 \right\rfloor$ envious agents. Let $X \subseteq \AA$ defined by $X = \{a \in \AA \mid \phi(a) \in \HH^*\}$. Therefore, $|X| = t = \frac{n}{2} - \left\lfloor \frac k 2 \right\rfloor$. Let $J \subseteq \AA$ be the set of envious agents. By assumption, $|J| \le 2 \left\lfloor \frac k 2 \right\rfloor$. Let $Y = \AA \setminus (X \cup J)$, be the set of non-envious agents outside $X$.
    
    Observe that no edge in $E_\AA$ can connect an agent in $X$ to an agent in $Y$; otherwise the corresponding agent in $Y$ would be envious of the corresponding agent in $X$. Note that $|Y \cup J| = \frac n 2 + \left\lfloor \frac k 2 \right\rfloor \ge 2 \left\lfloor \frac k 2 \right\rfloor$. Let $J'$ be an arbitrary subset of $Y \cup J$ and a superset of $J$ satisfying $|J'| = 2 \left\lfloor \frac k 2 \right\rfloor$. Then $J'$ is separates $V = \AA$ into $X$ and $(Y \setminus J')$ satisfying $|X| = |Y \setminus J'| = \frac n 2 - \left\lfloor \frac k 2 \right\rfloor$. Therefore, $(G, 2\left\lfloor \frac k 2 \right\rfloor)$ is a Yes-instance of the $1/2$-Vertex Separator problem, and so is $(G,k)$ (by Claim~\ref{clm:half-vs-equiv}).
    
    \textbf{Forward direction.}    
    On the other hand, if $(G, k)$ is a Yes-instance for the $1/2$-Vertex Separator problem, then so is $(G, 2 \left\lfloor \frac k 2 \right\rfloor)$ (by Claim~\ref{clm:half-vs-equiv}). By Claim~\ref{clm:half-exact}, there exists a partition $S, X, Y$ of $\AA$ such that $|S| = 2 \left\lfloor \frac k 2 \right\rfloor$, $|X| = |Y| = \frac n 2 - \left\lfloor \frac k 2 \right\rfloor$, and there is no edge between $X$ and $Y$ in $\GG_\AA$. We create an allotment $\phi$ defined as follows:
    \begin{itemize}
        \item Assign all houses in $\HH^*$ to agents in $X$ (arbitrarily).
        \item Assign all other houses to all other agents (arbitrarily).
    \end{itemize}
    The only envious agents can be the ones in $S$. Hence there are at most $2 \left\lfloor \frac k 2 \right\rfloor$ many envious agents. Therefore, the \ohaa instance $(\AA, \HH, \GG_\AA(\AA, E_\AA), (\PP_a)_{a \in \AA}, 2 \left\lfloor \frac k 2 \right\rfloor)$ is a Yes-instance.

    This completes the correctness of the reduction.
\end{proof}
}

With this, we derive that even when restricting to certain simple graph classes the problem of \ohaa might be harder than when restricting to just complete graphs.





\section{Exact Algorithms}\label{sec:exact}

In this section, we design exact algorithms for \ohaa. The naive brute force algorithm that enumerates all possible assignments of $m$ houses to $n$ agents takes time $2^{\OO(n \log m)}$. We design exact algorithms that are more efficient than this. Each algorithm caters to different graph properties of the input agent graphs and therefore become most efficient in different scenarios depending on the input agent graph. All the algorithms extend easily to work for \ohaah as well. 
First, we design an algorithm that is single exponential in the number of edges in the agent graph. This algorithm is more efficient that the above naive algorithm when the agent graph is a sparse graph. We also design an alternate separator-based algorithm that has better running time for graph classes with small balanced separators, like the class of planar graphs and bounded genus-graphs.

{\color{purple}



}
\subsection{Single Exponential Exact Algorithms}

First, we propose an exact algorithm with running time being single exponential in the number of nodes and edges.

\begin{theorem}\label{thm:exp-edge}
    There exists an algorithm that solves \ohaa for the instance $(\AA, \HH, \GG_\AA(\AA, E_\AA), (\PP_a)_{a \in \AA})$ in time $2^{n + 2|E_\AA|} \cdot (n + m)^{\OO(1)}$, where $n = |\AA|$ and $m = |\HH|$.
\end{theorem}
\begin{proof}
Consider the instance $(\AA, \HH, \GG(\AA, E_\AA), (\PP_a)_{a \in \AA})$. Agent $a$ is envious on agent $a'$ only if $\{a, a'\}$ is an edge in $E_\AA$. The flow of our algorithm will be as follows: (i) guess which agent is envious of which other agent, (ii) among the non-envious agents guess which subset of agent gets their preferred houses, (iii) use a bipartite matching algorithm to decide if such a allotment exists which matches our guesses.

For an allotment $\phi$, recall that $\EE_\phi^\HH(a)$ is the set of agents $a'$ where $a$ is envious on $a'$. If $N_{\GG_\AA}(a)$ is the set of neighbouring agents of $a$, then $\EE_\phi^\HH(a) \subseteq N_{\GG_\AA}(a)$. Since $|N_{\GG_\AA}(a)| = \deg(a)$ is the degree of $a$ in $\GG_\AA$, there are $2^{\deg(a)}$ possible choices of $\EE_\phi^\HH(a)$. We make a guess of $\EE_\phi^\HH(a)$ for all $a \in \AA$.

The agents with $\EE_\phi^\HH(a) = \emptyset$ are non-envious. A fraction of these agents receive their preferred houses, while all other agents receive a house not preferred by them. We make a guess $C \subseteq \{a \in \AA \mid \EE_\phi^\HH(a) = \emptyset\}$ of the agents who receive their preferred houses. There are at most $2^n$ possible guesses of $C$.

We now fix $C$ and $(\EE_\phi^\HH(a))_{a \in \AA}$ and we wish to check if there are any allotments which matches such a guess. We maintain a set $\FF_a \subseteq \HH$ for each $a \in \AA$ which would denote the houses that can be assigned to $a$ such that it matches the guess of $C$ and $(\EE_\phi^\HH(a))_{a \in \AA}$. We compute $\FF_a$ as follows.

\begin{itemize}
    \item We start with $\FF_a = \HH$ and iteratively trim $\FF_a$ by imposing restrictions, one by one.
    \item If $a \notin C$, then $\phi(a) \notin \PP_a$. Thus $\FF_a$ cannot contain any house in $\PP_a$. While, for $a \in C$, $\FF_a$ can only contain houses in $\PP_a$.
    \item For an edge $\{a, a'\}$ in $E_\AA$, if $a \in \EE_\phi^\HH(a')$, then $\phi(a') \notin \PP_{a'}$ but $\phi(a) \in \PP_{a'}$. Thus $\FF_a$ can contain only houses in $\PP_{a'}$. Note that we do not need to impose any further constraints on $\FF_{a'}$ as the previous step already trimmed $\FF_{a'}$ owing to $a \notin C$.
    \item Otherwise, for an edge $\{a, a'\}$ in $E_\AA$, if $a \in \EE_\phi^\HH(a')$ and $a' \notin C$, then $\phi(a') \notin \PP_{a'}$ and $\phi(a) \notin \PP_{a'}$. Therefore, $\FF_{a}$ cannot contain any house in $\PP_{a'}$.
    \item Finally, if for $\{a,a'\} \in E_\AA$, if $a \notin \EE_\phi^\HH(a')$ and $a' \in C$, then $\phi(a') \in \PP_{a'}$. Any house can be assigned to $a$ without making $a'$ envious; this imposes no further restrictions on $\FF_a$.
\end{itemize}

Note that, by construction of $(\FF_a)_{a \in \AA}$, the set of allotments $\phi$ that correspond to the guesses $C$, and $(\EE_\phi^\HH(a))_{a \in \AA}$ are exactly the set of allotments $\phi$, satisfying $\phi(a) \in \FF_a$. Since all such $\phi$ have the same envy, $\xi = |\{a \in \AA \mid \EE_\phi^\HH(a) \ne \emptyset\}|$, it suffices to check if at least one $\phi$ exists that correspond to the guesses of $C$ and $(\EE_\phi^\HH(a))_{a \in \AA}$. This check can be done using any maximum bipartite matching algorithm on the graph defined by $(\FF_a)_{a \in \AA}$. 

We already argued the correctness of this algorithm, all that remains is to argue its running time. For a fixed guess of $(\EE_\phi^\HH(a))_{a \in \AA}$ and $C$, the algorithm runs in polynomial time, including the maximum bipartite matching subroutine. There are at most $2^n$ guesses of $C$, while there are $2^{\deg(a)}$ guesses of $\EE_\phi^\HH(a)$. Therefore, the total number of guesses of $(\EE_\phi^\HH(a))_{a \in \AA}$ and $C$ are,

$$2^n \cdot \prod \limits_{a \in A} 2^{\deg(a)} = 2^{n + 2|E_\AA|}$$

Hence the algorithm terminates with a correct output in time $2^{n + 2|E_\AA|} \cdot (n + m)^{\OO(1)}$, where $n = |\AA|$ and $m = |\HH|$.
\end{proof}

We name this as Algorithm~\ref{alg:exp-edge}, the pseudocode of which is given below. 
Note that this algorithm is more efficient than the brute force for graphs with $|E_\AA| = o(n \log m)$, these include sparse graph classes like planar graphs, bounded treewidth graphs and bounded genus graphs.

\begin{algorithm}[!htbp]
    \caption{ \hfill \textbf{Input:} $\AA, \HH, \GG_\AA(\AA,E_\AA), (\PP_a)_{a \in \AA}$ \hfill \textbf{Output:} Minimum envy}\label{alg:exp-edge}
    \begin{algorithmic}[1]
        \State{$\xi \gets n$} \Comment{In the worst case, all agents are envious}
        \For{$(\EE_\phi^\HH)_{a \in \AA} \in (2^{N_{\GG_\AA}(a)})_{a \in \AA}$} \Comment{Guess $\EE_\phi^\HH(a) \subseteq N_{\GG_\AA}(a)$ for all $a \in \AA$}
            \For{$C \subseteq \{a \in \AA \mid \EE_\phi^\HH(a) = \emptyset\}$} \Comment{Guess $C \subseteq \AA$, agents getting preferred house}
                \For{$a \in \AA$}
                    \If{$a \in C$}
                        \State{$\FF_a \gets \PP_a$}
                    \Else
                        \State{$\FF_a \gets \HH \setminus \PP_a$}
                    \EndIf
                \EndFor
                \For{$\{a,a'\} \in E_\AA$}
                    \If{$a \in \EE_\phi^\HH(a')$} \Comment{Trim $\FF_a$}
                        \State{$\FF_a \gets \FF_a \cap \PP_{a'}$}
                    \ElsIf{$a' \notin C$}
                        \State{$\FF_a \gets \FF_a \setminus \PP_{a'}$}
                    \EndIf
                    \If{$a' \in \EE_\phi^\HH(a)$} \Comment{Repeat for $a'$}
                        \State{$\FF_{a'} \gets \FF_{a'} \cap \PP_{a}$}
                    \ElsIf{$a \notin C$}
                        \State{$\FF_{a'} \gets \FF_{a'} \setminus \PP_{a}$}
                    \EndIf
                \EndFor
                \State{$E_\FF \gets \{\{a, h\} \mid a \in \AA, h \in \FF_a$\}; $\GG_\FF \gets (\HH \cup \AA, E_\FF)$} \Comment{Feasibility bipartite graph}
                \If{maximum cardinality bipartite matching of $\GG_\FF$ is of size $|\AA|$}
                    \State{$\xi \gets \min(\xi, |\{a \in \AA \mid \EE_\phi^\HH(a) \ne \emptyset\}|$)}
                \EndIf
            \EndFor
        \EndFor 
        \State{\Return $\xi$}
    \end{algorithmic}
\end{algorithm}

\begin{remark}
    Algorithm~\ref{alg:exp-edge} can be trivially adapted to output the actual allocation in the same asymptotic runtime. Moreover, the algorithm can also be adapted to solve \ohaah; this can be done by maximizing $C$.
\end{remark}

\subsection{Balanced Separator Based Algorithm}

The notion of envy is intrinsically linked to the underlying graph structure $\GG_\AA$ and the exact edges play an important role in contributing to the total envy of any allotment. Two agents which do not share an edge in $E_\AA$ can never be envious of each other. So, it makes sense to look into separators of the graph $\GG_\AA$ which separate $\AA$ into nearly equal sized $\AA_1$ and $\AA_2$ which are independent of each other in the problem.

We start of with a definition of \emph{$f(n)$-balanced-separable} graph classes.

\begin{definition}
    A graph class $\mathfrak{G}$ is $f(n)$-balanced-separable, if for all graphs $\GG \in \mathfrak{G}$ of $n$ vertices, 
    \begin{itemize}
        \item all induced subgraphs of $\GG$ are also contained in $\mathfrak{G}$,
        \item there exists a partition $(S, V_1, V_2)$ of the vertex set of $\GG$ such that $|S| \le f(n)$, $|V_1|, |V_2| \le 2n/3$ and there is no edge in $\GG$ with one endpoint in $V_1$ and the other in $V_2$.
    \end{itemize}
\end{definition}

For example, the class of planar graphs are $\OO(\sqrt{n})$-balanced-separable (due to the planar separator theorem~\cite{lipton1979separator}), graphs of treewidth $\tw$ are $(\tw+1)$-balanced-separable.

We look into the problem of \ohaa when the underlying graph $\GG_\AA (\AA, E_\AA)$ is from an $f(n)$-balanced-separable graph class. Let $S \subseteq \AA$ be a subset of agents of size $f(n)$, such that $\AA \setminus S$ can be partitioned into $\AA_1, \AA_2$ where $|\AA_1|, |\AA_2| \le 2n/3$ and there is no edge between $\AA_1$ and $\AA_2$. Thus, the exact allotment of houses to agents in $\AA_2$ would not affect which agents in $\AA_1$ would be envious. This would allow us to recurse into independent subproblems with $\AA_1$ and $\AA_2$ as the set of agents.

We define a harder version of \ohaa, called \ohaab. Firstly, we allow the input to specify a feasibility set $\FF_a$ for each agent $a \in \AA$ which asks to only consider allotments $\phi$ satisfying $\phi(a) \in \FF_a$. Secondly, there is a provision to mark agents `angry'; an angry agent becomes envious whenever they don't get a house they prefer (irrespective of what their neighbours are assigned). The input specifies a subset $B \subseteq \AA$ of agents such that an agent $a \in B$ is envious if $\phi(a) \notin \PP_a$. As usual, for all other agents $a \in \AA \setminus B$, the agent $a$ is envious if $\phi(a) \in \PP_a$, and there exists $\{a, a'\} \in E_\AA$ such that $\phi(a') \in \PP_a$. 

\defproblem{\ohaab}{$\AA, \HH, \GG_\AA(\AA, E_\AA), (\PP_a)_{a \in \AA}, (\FF_a)_{a \in \AA}, B$.}{Compute an allocation $\phi$ satisfying $\phi(a) \in \FF_a$ which minimizes 
$$|\{a \in \AA \setminus B \mid \phi(a) \notin \PP_a, \exists a' \in \AA, \phi(a') \in \PP_a\}| + |\{a \in B \mid \phi(a) \notin \PP_a\}|$$} 

\begin{remark}
    \ohaab becomes the problem \ohaa when $B = \emptyset$, and $\FF_a = \HH$ for all $a \in \AA$.
\end{remark}

We now propose a recursive algorithm to solve \ohaab using balanced separators of size $f(n)$. 

\begin{theorem}\label{thm:sep-rec}
    There exists an algorithm that solves \ohaab, and hence \ohaa for the instance $(\AA, \HH, \GG_\AA(\AA, E_\AA), (\PP_a)_{a \in \AA})$ in time $2^{m + \OO(f(n) \cdot \log^2 n + n)}$ where $n = |\AA|$ and $m = |\HH|$ when the input graph is from a $f(n)$-balanced-separable graph class.
\end{theorem}

\begin{proof}
Let $(\AA, \HH, \GG_\AA(\AA, E_\AA), (\PP_a)_{a \in \AA})$ be an instance of \ohaa. Firstly, we only look into the case where the number of houses is same as the number of agents, i.e. $n = m$. If that is not the case, we can guess which houses appear in the final allotment, which blows up the final algorithm by a factor of at most $2^m$. 

Let $S \subseteq \AA$ be any subset of agents of size $f(n)$, such that $\AA \setminus S$ can be partitioned into $\AA_1, \AA_2$ where $|\AA_1|, |\AA_2| \le 2n/3$ and there is no edge between $\AA_1$ and $\AA_2$. $S$ can be computed naively using a brute force algorithm that enumerates all $\binom{n}{f(n)}$ subsets of size $f(n)$ and checking which of them form balanced separators. 

Our next step is to guess the houses that are allocated to agents in $S$, i.e. we guess $\phi(a) \in \FF_a$ for all $a \in S$. There are at most $n^{f(n)}$ many such guesses. We denote by $\phi(S)$ the set $\{\phi(a) \mid a \in S\}$. For a subset of agents in $S$, $\phi$ might already assign them their preferred houses, call this $C = \{a \in S \mid \phi(a) \in \PP_a\}$. Some agents in $S$ which are not angry (i.e. not in $B$) might be already envious of some other agent in $S$, call this set $D = \{a \in S \setminus B \mid \phi(a) \notin \PP_a, \exists a' \in S \cap N_{\GG_\AA}(a), \phi(a') \in \PP_a\}$. Some other agents who were angry might become envious from not getting a preferred house, call this set $Q = \{a \in S \cap B \mid a \notin \PP_a\}$. Let the remaining agents in $S$ be $R = S \setminus (C \cup D \cup Q)$.

The agents in $C$ are non-envious, and the agents in $D \cup Q$ are envious regardless of the allocation of the agents in $\AA_1 \cup \AA_2$. However, the agents in $R$ can be either envious or non-envious depending on what houses are allocated to which agents in $\AA_1 \cup \AA_2$. Note that $B \cap S \subseteq C \cup Q$, therefore, $R$ does not contain any angry agent in $A \in B$. Agents $a$ in $R$ got assigned some house other than their preference, so they are envious if and only if any of their neighbour in $\AA_1 \cup \AA_2$ gets assigned a house in $\PP_a$. 

We guess a subset $K \subseteq R$ to be non-envious in the final allocation. There are at most $2^{f(n)}$ such guesses. If an agent $a$ in $K$ shares an edge of $E_\AA$ with an agent $a'$ in $\AA_1 \cup \AA_2$, then $\phi(a') \notin \PP_a$. To model this, we trim the feasible set $\FF_{a'}$ to $\FF_{a'} \setminus \PP_a$. Call this modified feasibility set $(\FF'_a)_{a \in \AA_1 \cup \AA_2}$. Note that even agents in $R \setminus K$ could turn out to be non-envious, but this does not affect the correctness of the algorithm as some other guess of $K$ will match the exact subset of envious agents in a minimum envy allotment $\phi$. 

Let $a'$ be an agent in $\AA_1 \cup \AA_2$. Assume that there is an agent $a \in S$ such that the guessed allocation $\phi(a)$ is in $\PP_{a'}$. Therefore, $a'$ becomes envious if and only if $a'$ does not get a house in $\PP_{a'}$. We can therefore mark $a'$ as `angry' and include it in the set $B$, call this new set $B'$. This allows us to delete the information about the exact allocation in $S$ and recurse into independent subproblems in $\AA_1$ and $\AA_2$. To do this, we guess a subset of houses $\HH_1 \subseteq (\HH \setminus \phi(S))$ of size $|\AA_1|$ which are allocated to agents in $\AA_1$. There are $\binom{n}{|\AA_1|}$ many such guesses. Then $\HH_2 = \HH \setminus (\phi(S) \cup \HH_1)$ is the set of houses allocated to $\AA_2$. We update $\PP_a$ to $\PP_a \cap \HH_1$ if $a \in \AA_1$, or to $\PP_a \cap \HH_2$ if $a \in \AA_2$; call this updated preferences as $\PP'_a$. Let us denote by $\GG_\AA[\AA_1]$ the induced graph of $\GG_\AA$ by the set $\AA_1$, and similarly for $\GG_\AA[\AA_2]$.


Let `OHAANR' be a procedure that outputs the minimum envy for the problem of \ohaab. If $\zeta$ is the minimum envy corresponding to the guesses made by the algorithm so far, then we get the following relation:

\begin{align*}
\zeta =& |D| + |Q| + |R \setminus K| & [\text{envious agents in } S]\\ 
&+ \text{OHAANR}(\AA_1, \HH_1, \GG_\AA[\AA_1], (\PP'_a)_{a \in \AA_1}, (\FF'_a)_{a \in \AA_1}, B' \cap \AA_1) & [\text{envious agents in } \AA_1]\\
&+ \text{OHAANR}(\AA_2, \HH_2, \GG_\AA[\AA_2], (\PP'_a)_{a \in \AA_2}, (\FF'_a)_{a \in \AA_2}, B' \cap \AA_2) & [\text{envious agents in } \AA_2]
\end{align*}

Moreover, the minimum envy for the problem of \ohaab in the original instance $(\AA, \HH, \GG_\AA, (\PP_a)_{a \in \AA}, (\FF_a)_{a \in \AA}, B)$ is basically the minimum of $\zeta$ over all guesses of $(\phi(a))_{a \in S}$, $\HH_1$ and $K$.

$$\text{OHAANR}(\AA, \HH, \GG_\AA, (\PP_a)_{a \in \AA}, (\FF_a)_{a \in \AA}, B) = \min \limits_{(\phi(a))_{a \in S}, \HH_1, K} \big\{ \zeta \big\}$$

This allows us to design a recursive algorithm. The correctness of this algorithm is trivial from its design, as it recursively enumerates all possible allocations. We now analyze its running time.

Let $T(n)$ be the maximum running time that this algorithm takes on an input with $n$ agents and $m = n$ houses. Finding a balanced separator of size $f(n)$ takes time at most $\binom{n}{f(n)} \cdot n^2$. The total number of choices of $\phi(a)$ for $a \in S$ is at most $n^{f(n)}$. $\HH_1$ has $\binom{n}{|\AA_1|} \le 2^{n-1}$ choices. $K$ has $2^{f(n)}$ many choices. All other operations (except the recursion) take time polynomial in $n$. Therefore, they together take time:

$$\binom{n}{f(n)} \cdot (2n)^{f(n)} \cdot \binom{n}{|\AA_1|} \cdot n^{\OO(1)}= 2^{\OO(\log n) \cdot f(n) + (n - 1)}$$

After fixing $\phi(a)$ for $a \in S$, sets $K$ and $\HH_1$, the recursed instances take time at most $T(|\AA_1|)$ and $T(|\AA_2|)$. Since $|\HH_1| = |\AA_1| \le 2n/3$ and $|\HH_2| = |\AA_2| \le 2n/3$, we can upper bound both of these expressions by $T(2n/3)$. Moreover, $T(1) = n^{\OO(1)}$. This gives us the following recursive relation on the running time.

\begin{align*}
    T(n) &= 2^{\OO(\log n) \cdot f(n) + (n - 1)} \cdot 2 T(2n/3) = 2^{\OO(\log n) \cdot f(n) + n} \cdot T(2n/3)\\
    &= 2^{\OO(\log n) \cdot (f(n) + f(2n/3) + f(4n/9) + f(8n/27) + \cdots) + (n + 2n/3 + 4n/9 + 8n/27 + \cdots)} \\ 
    &= 2^{\OO(\log n) \cdot (f(n) + f(2n/3) + f(4n/9) + f(8n/27) + \cdots) + 3n}
\end{align*}

For arbitrary $f(n)$, we can bound $(f(n) + f(2n/3) + f(4n/9) + f(8n/27) + \cdots)$ by $f(n) \OO(\log n)$. This gives us $T(n) = 2^{\OO(f(n) \cdot \log^2 n + n)}$. As mentioned earlier, for $m \ne n$, the running time would become $2^{m + \OO(f(n) \cdot \log^2 n + n)}$.
\end{proof}

The pseudocode of this algorithm, named Algorithm~\ref{alg:sep-rec} is given below.

\begin{algorithm}[!htbp]
    \caption{ \hfill \textbf{Input:} $\AA, \HH, \GG_\AA(\AA,E_\AA), (\PP_a)_{a \in \AA}, (\FF_a)_{a \in \AA}, B$ \hfill \textbf{Output:} Minimum envy}\label{alg:sep-rec}
    \begin{algorithmic}[1]
        \Procedure{OHAANR}{$\AA, \HH, \GG_\AA(\AA,E_\AA), (\PP_a)_{a \in \AA}, (\FF_a)_{a \in \AA}, B$}
            \If{$\AA = \emptyset$} \Comment{Agent set is empty}
                \State{\Return $0$} \Comment{No agent is envious}
            \EndIf
            \State{$\xi \gets \infty$} \Comment{In the worst case, no assignment exists respecting $(\FF_a)_{a \in \AA}$}
            \State{$S \gets$ a balanced separator of size $f(n)$}
            \State{$\AA_1, \AA_2 \gets$ the separated parts of $\AA \setminus S$} \Comment{$|\AA_1|, |\AA_2| \le 2n/3$}
            \For{all choices of $\phi(a) \in \FF_a$ over all $a \in S$}
                \State{$C \gets \{a \in S \mid \phi(a) \in \PP_a\}$} \Comment{Agents with preferred house}
                \State{$D \gets \{a \in S \setminus B \mid \phi(a) \notin \PP_a, \exists a' \in S \cap N_{\GG_\AA}(a), \phi(a') \in \PP_a\}$} \Comment{Envious, non-angry}
                \State{$Q \gets \{a \in S \cap B \mid a \notin \PP_a\}$} \Comment{Envious, angry agents}
                \State{$R \gets S \setminus (C \cup D \cup Q)$} \Comment{Other agents}
                \State{$B' \gets B \cup \{a' \in \AA \setminus S \mid \exists a \in S, \{a,a'\} \in E_\AA, \phi(a) \in \PP_{a'}\}$} \Comment{Updated angry agents}
                \For{$\HH_1 \subseteq \HH \setminus \phi(S)$}
                    \State{$\HH_2 \gets \HH \setminus (\phi(S) \cup \HH_1)$}
                    \For{$K \subseteq R$}
                        \For{$a \in \AA_1$}
                            \State{$\FF'_{a} \gets \FF_a \cap \HH_1$}
                            \State{$\PP'_{a} \gets \PP_a \cap \HH_1$}
                        \EndFor
                        \For{$a \in \AA_2$}
                            \State{$\FF'_{a} \gets \FF_a \cap \HH_2$}
                            \State{$\PP'_{a} \gets \PP_a \cap \HH_2$}
                        \EndFor
                        \For{$a \in K$, $a' \in \AA \setminus S$ such that $\{a,a'\} \in E_\AA$}
                            \State{$\FF'_{a'} \gets \FF'_{a'} \setminus \PP_a$}
                        \EndFor
                        \State{$\xi_1 \gets \text{OHAANR}(\AA_1, \HH_1, \GG_\AA[\AA_1], (\PP'_a)_{a \in \AA_1}, (\FF'_a)_{a \in \AA_1}, B' \cap \AA_1)$}
                        \State{$\xi_2 \gets \text{OHAANR}(\AA_2, \HH_2, \GG_\AA[\AA_2], (\PP'_a)_{a \in \AA_2}, (\FF'_a)_{a \in \AA_2}, B' \cap \AA_2)$}
                        \State{$\xi \gets \min(\xi, |D| + |Q| + |R \setminus K| + \xi_1 + \xi_2)$}
                    \EndFor
                \EndFor
            \EndFor 
            \State{\Return $\xi$}
        \EndProcedure
    \end{algorithmic}
\end{algorithm}

Again, the algorithm can be trivially backtracked to provide the exact allocation $\phi$. Moreover, for all assignments of minimum envy, if we wish to maximize the number of agents getting a preferred house, this could also be done simply by maximizing $C = \{a \in S \mid \phi(a) \in \PP_a\}$ among all guesses of $C$ that lead to the same envy. 

\begin{remark}
    Algorithm~\ref{alg:sep-rec} can be trivially adapted to output the actual allocation in the same asymptotic runtime. Moreover, the algorithm can also be adapted solve \ohaah.
\end{remark}

Now, consider inputs with the guarantee $m = \OO(n)$. For treewidth $\le \tw$ graphs, this algorithms has worst case running time of $2^{\OO(\tw \log^2 n + n)}$. For  
$f(n) = \OO\left(\frac{n}{\log n}\right)$, the algorithm has the following worst case runtime:

\begin{align*}
    T(n) &= 2^{\OO\left(\log n \cdot \left(\frac{n}{\log n} + \frac{2n/3}{\log(2n/3)} + \frac{4n/9}{\log(4n/9)} + \cdots\right) + n\right)} &\\
    &= 2^{\OO\left(\log n \cdot \left(\frac{n}{\log n} + \frac{2n/3}{(3/4)\log(n)} + \frac{4n/9}{(9/16)\log(n)} + \cdots\right) + n\right)} & \text{as } (3/4)^r \log n \le \log((2/3)^r n) \\
    &= 2^{\OO\left(\log n \cdot \left(\frac{n}{\log n} + \frac{2n/4}{\log n} + \frac{4n/16}{\log n} + \cdots\right) + n\right)} = 2^{\OO(n)} & 
\end{align*}

In particular, these include planar graphs, bounded genus graphs~\cite{gilbert1984separator}, and any proper minor closed graph classes~\cite{kawarabayashi2010separator}. 

\begin{corollary}
    Consider an input $(\AA, \HH, \GG_\AA(\AA, E_\AA), (\PP_a)_{a \in \AA})$ with $m$ houses and $n$ agents and with the guarantee that $m = \OO(n)$. Then Algorithm~\ref{alg:sep-rec} solves \ohaab and hence \ohaa with the following worst case running times,
    \begin{itemize}
        \item $2^{\OO(\tw \log^2 n + n)}$ for treewidth $\le \tw$ graphs.
        \item $2^{\OO(n)}$ for $\OO\left(\frac{n}{\log n}\right)$-balanced-separable graphs.
    \end{itemize}
\end{corollary}
\section{Complexities for various Vertex Cover sizes of the Agent Graph} \label{sec:vc}

Input graphs with small vertex covers are particularly interesting to explore. Such graphs frequently arise in modeling subordinate-supervisor relationships in infrastructures with limited number of supervisors. In this Section, we first design an exact algorithm for \ohaa and \ohaah with a running time of $\OO((2m)^k \cdot (n + m)^{\OO(1)})$, where $m$ is the number of houses, $n$ is the number of agents in the agent graph $\GG_\AA$ and $k$ is the size of a vertex cover of $\GG_\AA$. Note that this algorithm performs better than Algorithm~\ref{alg:exp-edge} when $k = o \left(\frac{n+|E_\AA|}{\log m} \right)$, where $|E_\AA|$ is the number of edges in the agent graph. This algorithm also has quasi-polynomial running time when the input agent graphs have polylogarithmic-sized vertex covers.
On the other hand, we also show that the problems are \NPH even when the agent graph is a split graph or a bipartite graph with a vertex cover of size $n^{\varepsilon}$, where $n$ is the number of agents, and $\varepsilon \in (0,1)$ is a constant.

Let the graph on agents $\AA$ be $\GG_\AA = (\AA, E_\AA)$. We look into a slice-wise polynomial (XP) algorithm solving \ohaa with the vertex cover number as the parameter. 

\begin{theorem}\label{thm:xp-vc}
    Consider an input $(\AA, \HH, \GG_\AA(\AA, E_\AA), (\PP_a)_{a \in \AA})$ with $m$ houses and $n$ agents. There exists an algorithm solving \ohaa in time $\OO((2m)^k \cdot (n + m)^{\OO(1)})$, where $k$ is the size of a vertex cover of $\GG_\AA$.
\end{theorem}

\begin{proof}
    Given a set of $n$ agents $\AA$, $\GG_\AA = (\AA, E_\AA)$, set of $m$ houses $\HH$, preference $(\PP_a)_{a \in \AA}$, assume that a vertex cover $S \subseteq \AA$ of $\GG_\AA$ is given, such that $|S| = k$ (if not, this can be computed trivially in $n^{k+\OO(1)}$ time). Consider the following algorithm.

    \begin{itemize}
        \item Guess $\phi(a)$, for all $a \in S$. (there are $\OO(m^k)$ such guesses).
        \item Compute $J \subseteq S$, the set of agents already envious due to $\phi(S)$. Precisely, 
            $$J = \{a \in S \mid \phi(a) \notin \PP_a \text{ and }\exists a' \in S \cap N_{\GG_\AA}(a), \phi(a') \in \PP_a\}$$
        \item Guess happy agents $C \subseteq S \setminus J$ to be the set of agents who will remain non-envious after assigning houses to agents in $\AA \setminus S$ (there are $\OO(2^k)$ such guesses).
        \item We define a weighted, complete, bipartite graph $\GG_\FF((\AA \setminus S) \cup (\HH \setminus \phi(S)),E_\FF,w)$ as follows.
        \begin{itemize}
            \item $(\AA \setminus S)$ and $(\HH \setminus \phi(S))$ are the two parts. $E_\FF = \{\{a,h\} \mid a \in (\AA \setminus S), h \in (\HH \setminus \phi(S))\}$.
            \item Set the weights $w(a, h)$ of the edge $\{a,h\}$ to model whether $a$ becomes envious on getting $h$ or not. This can be achieved by doing the following:
            \begin{itemize}
                \item $w(a, h)$ is set to $\infty$, when the assignment of $h \in \HH \setminus \phi(S)$ to $a \in \AA \setminus S$ is not consistent with the guess of happy agents $C$. This happens when there is some agent $a' \in C$, such that $a'$ is a neighbour of $a$ but $a'$ is envious of $a$. Formally, 
                $w(a, h) = \infty$ when there exists $a' \in (C \cap N_{\GG_\AA}(a))$ such that $ \phi(a') \notin \PP_{a'}$ but $h \in \PP_{a'}$.
                \item $w(a, h)$ is set to $1$, when the assignment of $h$ to $a$ causes $a$ to become envious. Note that since we have already considered the allocation to be consistent with the guess of $C$, no agent in $C$ would be envious of $a$ getting $h$. This happens when there is some agent $a' \in S$, such that $a'$ is a neighbour of $a$ gets a house preferred by $a$ but $h$ is not preferred by $a$. Formally, 
                $w(a, h) = 1$ when there exists $a' \in (S \cap N_{\GG_\AA}(a))$ such that $ \phi(a') \in \PP_a$ but $h \in \PP_a$.
                \item When none of these is the case, the assignment of $h$ to $a$ does not increase envy. We set $w(a, h) = 0$.
            \end{itemize}
        \end{itemize}
        \item Compute the minimum cost maximum matching (left exhausting) of $\GG_\FF$ and obtain the cost to be $\zeta$. This equates to the total number of envious agents in corresponding assignment of the agents in $\AA \setminus S$.
        \item For some guess of $C$, the total number of envious agents in $S$ is at most $(k - |C|)$, while the total number of envious agents in $\AA \setminus S$ is $\zeta$. Since we are iterating over all guesses of $C$, for any minimum-envy allocation, there would exist a guess $C$ such that $(k - |C|)$ is exactly the number of envious agents. Hence, to solve \ohaa, it suffices to output the minimum of $k - |C| + \zeta$ over all guesses.
    \end{itemize}

    This algorithm takes time $\OO((2m)^k \cdot (n + m)^{\OO(1)})$ and essentially enumerates all possible allocations of $S$ while finding a corresponding minimum-envy extension of that. 
\end{proof}

We refer to this as Algorithm~\ref{alg:vc-xp} and provide its pseudocode below.   
    
\begin{algorithm}[htbp]
    \caption{ \hfill \textbf{Input:} $\AA, \HH, \GG_\AA(\AA,E_\AA), (\PP_a)_{a \in \AA}$, $S$ \hfill \textbf{Output:} Minimum envy}\label{alg:vc-xp}
    \begin{algorithmic}[1]
        \State{$\xi \gets n$} \Comment{In the worst case, all agents are envious}
        \For{guess $(\phi(a))_{a \in S}$}
            \State{$J \gets \{a \in S \mid \phi(a) \notin \PP_a \text{ and }\exists a' \in S \cap N_{\GG_\AA}(a), \phi(a') \in \PP_a\}$} \Comment{already envious}
            \For{guess $C \subseteq S \setminus J$} \Comment{envious in $S$}
                \For{($a \in (\AA \setminus S)$, $h \in (\HH \setminus \phi(S)$))}
                    \If{there exists $a' \in (C \cap N_{\GG_\AA}(a))$ such that $ \phi(a') \notin \PP_{a'}$ but $h \in \PP_{a'}$}
                        \State{$w(a, h) \gets \infty$}
                    \ElsIf{there exists $a' \in (S \cap N_{\GG_\AA}(a))$ such that $ \phi(a') \in \PP_a$ but $h \in \PP_a$}
                        \State{$w(a, h) \gets 1$}
                    \Else
                        \State{$w(a, h) \gets 0$}
                    \EndIf
                \EndFor
                \State{$\GG_\FF \gets$ complete bipartite graph with parts $(\AA \setminus S)$ and $(\HH \setminus \phi(S))$ and weights $w$}
                \State{$\zeta \gets$ minimum cost maximum bipartite matching of $\GG_\FF$}
                \State{$\xi \gets \min(\xi, k - |C| + \zeta)$}
            \EndFor
        \EndFor 
        \State{\Return $\xi$}
    \end{algorithmic}
\end{algorithm}

In Algorithm~\ref{alg:vc-xp}, if we modify the weight function as 
$$w'(a,h) = \begin{cases}
    w(a,h) - \frac{1}{n+1}  & \text{if } h \notin P_a\\
    w(a,h) & \text{otherwise}
\end{cases}$$
we would find a minimum envy allotment that maximizes the number of agents getting a preferred house, solving \ohaah.

\begin{remark}
    Algorithm~\ref{alg:vc-xp} could be modified slightly to solve \ohaah with the same asymptotic running time.
\end{remark}

This immediately gives us a polynomial time algorithms for star graphs, as they have vertex cover of size $1$. Moreover this translates to a quasi-polynomial time algorithm for graphs having vertex covers of size $(\log n)^{\OO(1)}$; hence we do not expect the problem to be \NPH for such graphs.

This motivates us to explore the exact boundary of tractable and intractable cases for various sizes of the minimum vertex cover. A natural question would be to ask whether vertex covers of size at most $n^\varepsilon$ for any constant $\varepsilon \in (0,1)$ would also allow efficient algorithms. We show in the next couple of theorems, that this is indeed not the case; \ohaa in \NPH even for bipartite graphs and split graphs with a vertex cover size of at most $n^\varepsilon$ for every constant $\varepsilon \in (0,1)$. Inspired from the reductions due to Madathil et. al.~\cite{madathil2024cost}, the reductions we provide work even when each house is preferred by constantly many agents.
\shortversion{
\begin{theorem}[$\star$]\label{thm:vc-bip}
    Let $\varepsilon \in (0, 1)$ be any constant. Consider an input $(\AA, \HH, \GG_\AA(\AA, E_\AA), (\PP_a)_{a \in \AA})$ with $m$ houses and $n$ agents. \ohaa is \NP-hard even when $\GG_\AA$ is a complete bipartite graph with one of its partition being of size at most $n^\varepsilon$, and each house is preferred by at most four agents.
\end{theorem}
}
\longversion{
\begin{theorem}\label{thm:vc-bip}
    Let $\varepsilon \in (0, 1)$ be any constant. Consider an input $(\AA, \HH, \GG_\AA(\AA, E_\AA), (\PP_a)_{a \in \AA})$ with $m$ houses and $n$ agents. \ohaa is \NP-hard even when $\GG_\AA$ is a complete bipartite graph with one of its partition being of size at most $n^\varepsilon$, and each house is preferred by at most four agents.
\end{theorem}

\begin{proof}
    We reduce from the \NP-hard problem of \CLIQUE. Let $(G(V,E), k)$ be an instance of \CLIQUE, where the problem asks to decide if $G$ has a clique of size $k$. Let $|V| = N$ and $|E| = M$. Moreover, we can assume $N \ge (2M)^{\lceil 1 / \varepsilon \rceil}$; indeed if that is not the case, we can add (in polynomial time), $(2M) ^ {\lceil 1 / \varepsilon \rceil}$ many isolated vertices to $G$ without changing the size of the maximum clique.
    
    We reduce $(G(V,E), k)$ to an instance of \ohaa as follows. 

    \begin{itemize}
        \item Let $\AA = \{a_v \mid v \in V\} \cup \{a^1_e \mid e \in E\} \cup \{a^2_e \mid e \in E\}$. That is, we have an agent $a_v$ for every vertex $v$ of $G$, and two agents $a^1_e$, and $a^2_e$ for each edge $e$ of $G$. Therefore, $n = |\AA| = N + 2M$. 
        \item Let $\HH = \{h_e \mid e \in E\} \cup \{h^*_j \mid j \in [N+2M-\binom{k}{2}]\}$. Therefore, $|\HH| = M + N + 2M - \binom{k}{2} = |\AA| + (M-\binom{k}{2})$. 
        \item We are now going to set the houses $h_e$ to be preferred by agents $a^1_e$ and $a^2_e$, and also by $a_u$ and $a_v$ where $e = \{u,v\}$. This is effectively done by setting $\PP_{a^1_e} = \PP_{a^2_e} = \{h_e\}$, and $\PP_{a_v} = \{h_{e'} \mid e' \in E \text{ and } v \in e'\}$. The houses $\{h^*_j \mid j \in [N+2M - \binom{k}{2}]\}$ are preferred by no agents, i.e. these are dummy houses.
        \item We now define the underlying graph $\GG_\AA(\AA, E_\AA)$. Let $\AA_V = \{a_v \mid v \in V \}$ and let $\AA_E = \{a^1_e \mid e \in E\} \cup \{a^2_e \mid e \in E\}$. Then $\GG_\AA$ defines the complete bipartite graph between $\AA_V$ and $\AA_E$. That is, $E_\AA = \{\{a,a'\} \mid a \in \AA_V, a' \in \AA_E\}$. 
        \item The target number of envious agents is $k$.
    \end{itemize}
    Notice that the part $\AA_E$ of the complete bipartite graph $\GG_\AA$ satisfies $|\AA_E| = 2M \le N^\varepsilon = |\AA_V|^\varepsilon \le |\AA|^\varepsilon$. This reduction can be done in polynomial time. We now show the correctness of the reduction.

    \textbf{Forward direction.} 
    Let $(G(V, E), k)$ be a Yes-instance of \CLIQUE, i.e., there is clique $S \subseteq V$ of size $k$. Let $E_S$ be the edges of the clique $S$, $|E_S| = \binom{k}{2}$. Consider the allocation $\phi$ defined as:
    \begin{itemize}
        \item A dummy house is assigned to $a_v$ for all $v \in V$. This uses $N$ dummy houses.
        \item A dummy house is assigned to $a^2_e$ for all $e \in E$. This uses $M$ dummy houses.
        \item A dummy house is assigned to $a^1_e$ for every $e \in E \setminus E_S$. This uses $(M - \binom{k}{2})$ dummy houses. 
        \item For every $e \in E_S$,the house $h_e$ is assigned to $a^1_e$. Hence the houses $\{h_e \mid e \in E \setminus E_S\}$ are all unallocated.
    \end{itemize}

    Note that for every $e \in E$, the agents $a^1_e$ and $a^2_e$ are all non-envious, because their neighbours in $\GG_\AA$, i.e. the agents $\{a_v \mid v \in V\}$ are all assigned dummy houses. Moreover, for all $v \in V$, such that $v \notin S$, no edge in $E_S$ is incident to $v$; thus all houses preferred by $a_v$ are unallocated, making $a_v$ non-envious. Therefore, all envious agents are in $\{a_v \mid v \in S\}$. This makes the number of envious agents in $\phi$ to be at most $|S| = k$; the \ohaa instance is a Yes-instance.

    \textbf{Reverse direction.}
    Let $\phi$ be some allocation with at most $k$ envious agents. We say that $\phi$ is nice if for all $e \in E$, $h_e$ is either unassigned or is assigned to $a^1_e$. If $\phi$ is not nice to begin with, then $\phi$ can be converted to a nice allocation by a sequence of exchanges, none of which increases total envy. 
    
    Let $h_e$ be some house which is neither unallocated, nor allocated to $a^1_e$. 
    
    \begin{itemize}
        \item \textbf{Case I: $h_e$ is assigned to $a^2_e$.} We swap the houses allocated to $a^1_e$ and $a^2_e$. This does not change the number of envious agents, due to the symmetry of $\GG_\AA$.
        \item \textbf{Case II: $h_e$ is assigned to $a^j_{e'}$ for some $j \in \{1,2\}$, $e' \ne e$.} We swap the houses allocated to $a^1_e$ and $a^j_{e'}$. Again, this does not increase the number of envious agents. This is because
        \begin{enumerate}
            \item $a^1_e$ is not envious in both allocations.
            \item $a^j_{e'}$ did not have their preferred house before the swap. Since the neighbourhood of $a^j_{e'}$ is unaffected by the swap, $a^j_{e'}$ cannot become envious after the swap, if $a^j_{e'}$ was not envious before the swap.
            \item For all other agents, they are envious after the swap if and only if they were envious before the swap. This is because the set of houses allocated to their neighbours remain unchanged.
        \end{enumerate}
        \item \textbf{Case III: $h_e$ is assigned to $a_v$ for some $v \in V$.} Agents $a^1_e, a^2_e$ must be envious of $a_v$ in $\phi$. We now look at the sequence of agents $a^{1}_{e_0} = a^1_e, a^{1}_{e_1}, \ldots, a^{1}_{e_{p-1}}, a^{1}_{e_p}$ such that 
        \begin{itemize}
            \item $\phi(a_v) = h_{e_0}$.
            \item $\phi(a^{1}_{e_0}) = h_{e_1}, \phi(a^{1}_{e_1}) = h_{e_2}, \ldots, \phi(a^{1}_{e_{p-1}}) = h_{e_p}$.
            \item $\phi(a^{1}_{e_p}) = h^*_q$ is a dummy house.
        \end{itemize}

        Note that such a sequence must exist and end with an agent who is assigned a dummy house, as the number of agents and houses are finite. Therefore, $p$ is a finite integer. Now, we modify the allocation to $\phi'$ as follows:

        \begin{itemize}
            \item $\phi'(a_v) = h^*_q$.
            \item $\phi'(a^{1}_{e_0}) = h_{e_0}, \phi'(a^{1}_{e_1}) = h_{e_1}, \ldots, \phi'(a^{1}_{e_{p-1}}) = h_{e_{p-1}}, \phi'(a^{1}_{e_p}) = h_{e_p}$.
            \item $\phi'(a) = \phi(a)$ for all other agents $a$.
        \end{itemize}

        We will now argue why this set of exchanges do not increase the number of envious agents.
        \begin{itemize}
            \item \textsl{Case A: $e$ is incident on $v$.} Let $e = \{u,v\}$. In $\phi$, $a_v$ was not envious, but $a_u$ may or may not be envious. The agents $a^1_e$, and $a^2_e$ were both envious. In $\phi'$, $a^1_e$ and $a^2_e$ are not envious, agent $a_v$ becomes envious, and $a_u$ may become envious. Therefore, the number of envious agents in $\{a^1_e, a^2_e, a_u, a_v\}$ do not increase. The agents $a^1_{e_1}, a^1_{e_2}, \ldots, a^1_{e_p}$ all become non-envious in $\phi'$ as they get their preferred houses. For all other agents $a^j_{e'}$ where $e'$ is an edge in $E$, the dummy house $h^*_q$ is the only house which appears in the set of houses assigned to neighbours of $a^j_{e'}$ in $\phi'$ which did not appear in $\phi$; hence $a^j_{e'}$ is envious in $\phi'$ only if they were envious in $\phi$. Finally, all other agents $a_{v'}$ where $v'$ is a vertex in $V$, are envious in $\phi'$ if and only if they are envious in $\phi$ as only the house $h_v$ (a house not preferred by $a_{v'}$) gets added to the set of houses assigned to the neighbours of $a_{v'}$. 
            \item \textsl{Case B: $e$ is not incident on $v$.} Let $e=\{u_1, u_2\}$. In $\phi$, the agents $a_{u_1}$, $a_{u_2}$ may or may not be envious; but the agents $a^1_e$ and $a^2_e$ were both envious. In $\phi'$, $a^1_e$ and $a^2_e$ are not envious, while agents $a_{u_1}$ and $a_{u_2}$ may become envious. Therefore, the number of envious agents in $\{a_{u_1}, a_{u_2}, a^1_e, a^2_e\}$ do not increase. The agents $a^1_{e_1}, a^1_{e_2}, \ldots, a^1_{e_p}$ all become non-envious in $\phi'$ as they get their preferred houses. If the agent $a_v$ is not envious in $\phi$, then it cannot become envious in $\phi'$ as only $h_{e}$ (a house not preferred by $a_v$) gets added to the set of houses assigned to its neighbours. For all other agents $a^j_{e'}$ where $e'$ is an edge in $E$, the dummy house $h^*_q$ is the only house which appears in the set of houses assigned to neighbours of $a^j_{e'}$ in $\phi'$ which did not appear in $\phi$; hence $a^j_{e'}$ is envious in $\phi'$ only if they were envious in $\phi$. Finally, all other agents $a_{v'}$ where $v'$ is a vertex in $V$, are envious in $\phi'$ if and only if they were envious in $\phi$ as only the house $h_v$ (a house not preferred by $a_{v'}$) gets added to the set of houses assigned to the neighbours of $a_{v'}$. 
        \end{itemize}
    \end{itemize}
    These exchanges do not increase the total envy, but they allocate at least one house $h_e$ to $a^1_e$, which was previously allocated to some other agent. Repeating this for at most $|\HH|$ times would give us a nice allocation, with envy no more than that in $\phi$.

    Hence, we can safely assume that $\phi$ is a nice allocation with envy at most $k$. $\phi$ allocates dummy houses to $a_v$ for all $v \in V$ and to $a^2_e$ for all $e \in E$. Moreover, it is safe to assume that all dummy houses are assigned to some agent, otherwise we can assign a dummy house to any $a^1_e$ without increasing total envy.

    We now look into agents which are not assigned dummy vertices by $\phi$. There are $N + 2M - (N + 2M - \binom{k}{2}) = \binom{k}{2}$ many such agents, each of the form $a^1_e$ for some $e \in E$. Let $\gamma = \binom{k}{2}$. Let $a^{1}_{e_1}, a^{1}_{e_2}, \ldots, a^{1}_{e_\gamma}$ be the agents which are not assigned dummy houses. For any $v \in V$, $v$ is an endpoint of at least one edge in $e_1, e_2, \ldots, e_\gamma$, if and only if the agent $a_v$ is envious. Therefore, the number of envious agents is precisely the number of vertices on which the edges $e_1, e_2, \ldots, e_\gamma$ are incident on. The number of vertices that $\binom{k}{2}$ distinct edges are incident on, is at least $k$, with equality holding if and only if the edges induce a clique of size $k$. Further, the number of envious agents, i.e. the number of vertices on which the edges $e_1, e_2, \ldots, e_\gamma$ are incident on, is at most $k$ by assumption. Therefore indeed, $\{e_1, e_2, \ldots, e_\gamma\}$ is a set of $\binom{k}{2}$ distinct edges, which induce a clique of size $k$. 

    This completes the proof of correctness of the reduction.
\end{proof}
}







Next we explore split graphs. Split graphs are graphs that can be partitioned into a clique and an independent set. Note that the clique can act as a vertex cover. We show the the problem \ohaa, even when restricted to split graphs of $n$-vertices having the clique of size at most $n^\varepsilon$, for any constant $\varepsilon \in (0,1)$, is \NPH. 

\begin{theorem}\label{thm:vc-split}
    Let $\varepsilon \in (0, 1)$ be any constant. Consider an input $(\AA, \HH, \GG_\AA(\AA, E_\AA), (\PP_a)_{a \in \AA})$ with $m$ houses and $n$ agents. \ohaa is \NP-hard even when $\GG_\AA$ is a split graph where the maximum clique is of size at most $n^\varepsilon$, and each house is preferred by at most three agents.
\end{theorem}

\begin{proof}
    We provide a reduction from the \CLIQUE~problem to \ohaa. Let $(G(V,E), k)$ be an instance of \CLIQUE, which asks to decide whether $G$ has a clique of size $k$. 

    Let $|V| = N, |E| = M \ge 1$. Define $t = N ^{\lceil 1 / \varepsilon \rceil}$. We now construct an instance of \ohaa as follows:
    
    \begin{itemize}
        \item Let $\AA = \{a_v \mid v \in V\} \cup \{a^j_e \mid e \in E, j \in [t]\}$. That is, we have an agent $a_v$ for every vertex $v$ of $G$, and $t$ agents $a^1_e, a^2_e, \ldots, a^t_e$ for each edge $e$ of $G$. Therefore, $n = |\AA| = N + M \cdot t = N + M \cdot N ^{\lceil 1/\varepsilon \rceil}$. 
        \item Let $\HH = \{h^j_e \mid e \in E, j \in [t]\} \cup \{h^*_j \mid j \in [(M - \binom{k}{2})t + N]\}$. Therefore, $|\HH| = M \cdot t + N + (M - \binom{k}{2})t = |\AA| + (M - \binom{k}{2})t$. 
        \item We are now going to set the houses $h^j_e$ to be preferred by agents $a^j_e$, $a_{u}, a_{v}$, where $e = \{u, v\}$. This is effectively done by setting $\PP_{a^j_e} = \{h^j_e\}$, and $\PP_{a_v} = \{h^j_e \mid j \in [t], e \in E \text{ such that } v \in e\}$. The houses $\{h^*_j \mid j \in [(M - \binom{k}{2})t + N]\}$ are preferred by no agents, i.e., these are dummy houses.
        \item We now define the underlying graph $\GG_\AA(\AA, E_\AA)$. The edges in $E_\AA$ are defined as follows:
        \begin{itemize}
            \item For every $u, v \in V$, $\{a_u, a_v\} \in E_\AA$.
            \item For every $v \in V, e \in E, j \in [t]$, $\{a_{v}, a^j_e\} \in E_\AA$.
        \end{itemize}
        \item The target number of envious agents is $k$.
    \end{itemize}

    This means that $C = \{a_v \mid v \in V\}$ forms a clique and $I = \{a^j_e \mid e \in E, j \in [t]\}$ forms an independent set with $\AA = C \cup I$. Moreover, $|C| = N < (N + M\cdot N^{\lceil 1 / \varepsilon \rceil})^\varepsilon = n^\varepsilon$. The output of the reduction is the \ohaa instance $(\AA, \HH, \GG_\AA, (\PP_a)_{a \in \AA}, k)$. This reduction can be done in polynomial time as $\varepsilon$ is a constant. We now show the correctness of this reduction.

    \textbf{Forward direction.}
    Let $(G(V, E), k)$ be a Yes-instance of \CLIQUE, i.e., there is a clique $S \subseteq V$ of size $k$. Let $E_S$ be the edges of the clique $S$. We have $|E_S| = \binom{k}{2}$. Consider the allocation $\phi$ defined as:
    \begin{itemize}
        \item A dummy house is assigned to $a_v$ for all $v \in V$. This uses $N$ dummy houses.
        \item A dummy house is assigned to $a^j_e$ for every $j \in [t], e \in E \setminus E_S$. This uses $(M - \binom{k}{2})t$ dummy houses. 
        \item For every $j \in [t], e \in E_S$, the house $h^j_e$ is assigned to $a^j_e$. Hence the houses $\{h^j_e \mid j \in [t], e \in E \setminus E_S\}$ are all unallocated.
    \end{itemize}

    Note that for every $e \in E$, the agents $a^j_e$ are all non-envious, because their neighbours in $\GG_\AA$, i.e., the agents $\{a_v \mid v \in V\}$ are all assigned dummy houses. Moreover, for all $v \in V$, such that $v \notin S$, no edge in $E_S$ is incident to $v$; thus, all houses preferred by $a_v$ are unallocated, making $a_v$ non-envious. Therefore, all envious agents are in $\{a_v \mid v \in S\}$. This makes the number of envious agents in $\phi$ to be at most $|S| = k$; the \ohaa instance is a Yes-instance. 

    \textbf{Reverse direction.}
    Let $\phi$ be some allocation with at most $k$ envious agents. We say that $\phi$ is nice if for all $j \in [t], e \in E$, $h^j_e$ is either unassigned or is assigned to $a^j_e$. If $\phi$ is not nice to begin with, then $\phi$ can be converted to a nice allocation by a sequence of exchanges, none of which increases total envy. 
    
    Let $h^j_e$ be some house which is neither unallocated, nor allocated to $a^j_e$. 
    
    \begin{itemize}
        \item \textbf{Case I: $h^j_e$ is allocated to $a^{j'}_{e'}$ for some $j' \ne j$ or $e' \ne e$.} Then we swap houses allocated to $a^j_e$ and $a^{j'}_{e'}$. After this swap, $a^j_e$ is non-envious by getting a preferred house. The agent $a^{j'}_{e'}$ cannot become envious after the swap if $a^{j'}_{e'}$ was not envious before the swap; this is because the houses allocated to the neighbours of $a^{j'}_{e'}$ remain unchanged and $a^{j'}_{e'}$ did not have a preferred house before the swap. For all other agents, they remain envious after the swap if and only if they were envious before the swap (as the set of houses allocated to their neighbours is unchanged).
        \item \textbf{Case II: $h^j_e$ is allocated to $a_v$ for some $v \in V$.} $a^j_e$ must be envious of $a_v$ in $\phi$. We now look at the sequence of agents $a^j_e = a^{j_0}_{e_0}, a^{j_1}_{e_1}, \ldots, a^{j_{p-1}}_{e_{p-1}}, a^{j_p}_{e_p}$ such that 
        \begin{itemize}
            \item $\phi(a_v) = h^{j_0}_{e_0}$.
            \item $\phi(a^{j_0}_{e_0}) = h^{j_1}_{e_1}, \phi(a^{j_1}_{e_1}) = h^{j_2}_{e_2}, \ldots, \phi(a^{j_{p-1}}_{e_{p-1}}) = h^{j_p}_{e_p}$.
            \item $\phi(a^{j_p}_{e_p}) = h^*_q$ is a dummy house.
        \end{itemize}
        Note that such a sequence must exist and end with an agent who is assigned a dummy house, as the number of agents and houses are finite. Therefore, $p$ is a finite integer. Now, we modify the allocation to $\phi'$ as follows:

        \begin{itemize}
            \item $\phi'(a_v) = h^*_q$.
            \item $\phi'(a^{j_0}_{e_0}) = h^{j_0}_{e_0}, \phi'(a^{j_1}_{e_1}) = h^{j_1}_{e_1}, \ldots, \phi'(a^{j_{p-1}}_{e_{p-1}}) = h^{j_{p-1}}_{e_{p-1}}, \phi'(a^{j_p}_{e_p}) = h^{j_p}_{e_p}$.
            \item $\phi'(a) = \phi(a)$ for all other agents $a$.
        \end{itemize}

        In $\phi'$, $a^{j}_e$ is not envious anymore, however, $a_v$ may or may not become envious. The agents $a^{j_1}_{e_1}, \ldots, a^{j_{p-1}}_{e_{p-1}}, a^{j_p}_{e_p}$, were allocated their preferred house in $\phi'$, and hence are not envious. Moreover, for all other $a^{j'}_{e'}$, they are envious in $\phi'$ if and only if they were envious in $\phi$; this is because the dummy house $h^*_q$ is the only house that appears in the set of houses allocated to their neighbours in $\phi'$ but not in $\phi$. The same happens for $a_{v'}$, $v' \in V \setminus \{v\}$; the set of houses allocated to the neighbours of $a_{v'}$ does not change. Therefore this does not increase total envy.
    \end{itemize}

    These exchanges do not increase the total envy, but they allocate at least one house $h^j_e$ to $a^j_e$, which was previously allocated to some other agent. Repeating this for at most $|\HH|$ times would give us a nice allocation, with envy no more than that in $\phi$.

    Hence, we can safely assume that $\phi$ is a nice allocation with envy at most $k$. $\phi$ allocates dummy houses to $a_v$ for all $v \in V$. Moreover, it is safe to assume that all dummy houses are assigned to some agent, otherwise we can assign a dummy house to any $a^j_e$ without increasing total envy.

    We now look into agents which are not assigned dummy vertices by $\phi$. There are $Mt - (M - \binom{k}{2})t = \binom{k}{2}t$ many such agents. Let $\gamma = \binom{k}{2}t$. Let $a^{j_1}_{e_1}, a^{j_2}_{e_2}, \ldots, a^{j_\gamma}_{e_\gamma}$ be the agents which are not assigned dummy houses. Therefore, the number of distinct edges in $e_1, e_2, \ldots, e_\gamma$ is at least $\gamma / t = \binom{k}{2}$. Moreover, for any $v \in V$, $v$ is an endpoint of at least one edge in $e_1, e_2, \ldots, e_\gamma$, if and only if the agent $a_v$ is envious. Therefore, the number of envious agents is precisely the number of vertices on which the edges $e_1, e_2, \ldots, e_\gamma$ are incident on. The number of vertices that $\binom{k}{2}$ distinct edges are incident on, is at least $k$, with equality holding if and only if the edges induce a clique of size $k$. However, the number of envious agents, i.e., the number of vertices on which the edges $e_1, e_2, \ldots, e_\gamma$ are incident on, is at most $k$ by assumption. Therefore, $\{e_1, e_2, \ldots, e_\gamma\}$ is a set of $\binom{k}{2}$ distinct edges, which induce a clique of size $k$. 

    This completes the proof of correctness of the reduction.
\end{proof}

\section{Conclusion}\label{sec:concl}

We study the \ohaa and \ohaah problem, natural extensions of the classical {\sc House Allocation} problem to social networks.

Our results reveal a rich complexity theoretic landscape. We present polynomial-time algorithms when each agent has a unique preferred house. Then, we show NP-hardness results under structural restrictions, including split graphs, complete bipartite graphs, and 3-regular graphs with uniform preferences. We further look into structural properties of the agent graphs and design exact (exponential) algorithms exploiting these structures, that are more efficient than the naive brute-force algorithm. 

An interesting problem for future work would be to analyze the following decision problem: Given a set of agents $\AA$, a set of houses $\HH$, an underlying agent graph $\GG_\AA$, preferences $(\PP_a)_{a \in \AA}$, and integers $k$, $b$, determine if there is an allocation that achieves envy at most $k$ and happiness at least $b$. Note that this is at least as hard as \ohaah.

\newpage
\bibliography{refs}

\newpage
\appendix
\shortversion{

\section{Proof of Theorem~\ref{thm:bip-d-2}}

\begin{figure} [ht]
    \begin{subfigure}[t]{0.95\textwidth} 
        \centering 
        \includegraphics[width=0.3\textwidth]{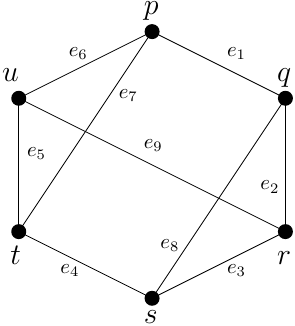}
        \caption{Instance of CLIQUE for $3$-regular graphs, $k = 3$}
    \end{subfigure} \\
    \begin{subfigure}[t]{0.95\textwidth}
        \centering
        \includegraphics[width=0.9\textwidth]{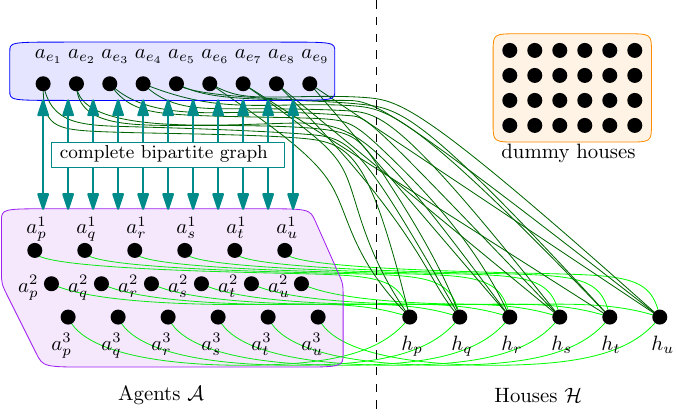}
        \caption{The reduced instance as per Theorem~\ref{thm:bip-d-2}}
    \end{subfigure} \hfill \hfill \hfill
    \caption{Reduction for Theorem~\ref{thm:bip-d-2}} \label{fig:bip-d-2}
\end{figure}

\begin{proof}
    We provide a reduction from the \NP-hard problem of \CLIQUE~on regular graphs to \ohaa. Let $(G(V,E), k)$ be an instance of \CLIQUE~on regular graphs where $G$ is $\delta$-regular and the problem asks to decide if $G$ has a clique of size $k$. Let $|V| = N$ and $|E| = M$. 

    We reduce $(G(V,E), k)$ to an instance of \ohaa as follows. 

    \begin{itemize}
        \item Let $\AA = \{a^j_v \mid v \in V, j \in [\delta]\} \cup \{a_e \mid e \in E\}$. That is, we have an agent $a_e$ for every edge $e$ of $G$, and $\delta$ agents $a^1_v, a^2_v, \ldots, a^\delta_v$ for each vertex $v$ of $G$. Therefore, $n = |\AA| = \delta N + M$. 
        \item Let $\HH = \{h_v \mid v \in V\} \cup \{h^*_j \mid j \in [\delta N+M-k]\}$. Therefore, $|\HH| = N + \delta N + M - k = |\AA| + (N-k)$. 
        \item We are now going to set the houses $h_v$ to be preferred by agents $a^j_v$ for all $j \in [\delta]$, and also by $a_e$, for all edges $e$ which are incident on $v$. This is effectively done by setting $\PP_{a^j_v} = \{h_v\}$, and $\PP_{a_e} = \{h_u, h_v\}$, where $e = \{u,v\}$. The houses $\{h^*_j \mid j \in [\delta N + M - k]\}$ are preferred by no agents, i.e., these are dummy houses.
        \item We now define the underlying graph $\GG_\AA(\AA, E_\AA)$. Let $\AA_V = \{a^j_v \mid v \in V, j \in [\delta]\}$ and let $\AA_E = \{a_e \mid e \in E\}$. Then $\GG_\AA$ defines the complete bipartite graph between $\AA_V$ and $\AA_E$. That is, $E_\AA = \{\{a,a'\} \mid a \in \AA_V, a' \in \AA_E\}$. 
        \item The target number of envious agents is $k\delta - \binom{k}{2}$. Not that for non-trivial instances, $k \le \delta$, hence $k\delta - \binom{k}{2}$ is a positive integer.
    \end{itemize}
    
    Please refer to Figure~\ref{fig:bip-d-2} for clarity. Every step can be done in polynomial time; hence this is a polynomial-time reduction. We now show the correctness of the reduction. 

    \textbf{Forward direction.}
    Let $S \subseteq V$ be a clique in $G$ of size $|S| = k$. Consider the allocation $\phi$ defined as follows:
    \begin{itemize}
        \item Assign $h_v$ to $a^1_v$ for all $v \in S$. 
        \item For the remaining $|\AA| - k = \delta N+M-k$ agents, assign them a dummy house each. 
    \end{itemize}
    
    In $\phi$ the houses $h_v$ are left unassigned for $v \notin S$. No agent in $\AA_V$ are envious as all their neighbours (i.e. $\AA_E$) are assigned dummy houses. Moreover, $a_e$ is not envious if $e$ is not incident on $S$; this is because the houses preferred by $a_e$ are unallocated. Therefore, the number of envious agents is at most the number of edges incident on $S$; this is precisely $k\delta - \binom{k}{2}$, as $G$ is $\delta$-regular and $S$ is a clique of size $k$.

    \textbf{Reverse direction.}
    Let $\phi$ be an allocation with at most $k\delta - \binom{k}{2}$ envious agents. We say that an allocation $\phi$ is nice if for all $v \in V$, the house $h_v$ is either unassigned, or assigned to the agent $a^1_v$. If $\phi$ is not nice to begin with, then $\phi$ can be converted to a nice allocation by a sequence of exchanges, none of which increase total envy.

    Let $h_v$ be a house that is assigned to an agent in $\AA$ other than $a^1_v$.

    \begin{itemize}
        \item \textbf{Case I: $h_v$ is assigned to $a^j_v$ for some $j \ne 1$.} We swap the houses allocated to $a^j_v$ and $a^1_v$. This does not change the number of envious agents, due to the symmetry of $\GG_\AA$.
        \item \textbf{Case II: $h_v$ is assigned to $a^j_u$ for some $j \in [\delta]$, $u \ne v$.} We swap the houses allocated to $a^j_u$ and $a^1_v$. Again, this does not increase the number of envious agents. This is because
        \begin{enumerate}
            \item $a^1_v$ is not envious in both allocations.
            \item $a^j_u$ did not have their preferred house before the swap. Since the neighbourhood of $a^j_u$ is unaffected by the swap, $a^j_u$ cannot become envious after the swap if $a^j_u$ was not envious before the swap.
            \item For all other agents, they are envious after the swap if and only if they were envious before the swap. This is because the set of houses allocated to their neighbours remain unchanged.
        \end{enumerate}
        \item \textbf{Case III: $h_v$ is assigned to $a_e$ for some $e \in E$.} $a^1_v, a^2_v \ldots, a^\delta_v$ must be envious of $a_e$ in $\phi$. We now look at the sequence of agents $a^{1}_{v_0} = a^1_v, a^{1}_{v_1}, \ldots, a^{1}_{v_{p-1}}, a^{1}_{v_p}$ such that 
        \begin{itemize}
            \item $\phi(a_e) = h_{v_0}$.
            \item $\phi(a^{1}_{v_0}) = h_{v_1}, \phi(a^{1}_{v_1}) = h_{v_2}, \ldots, \phi(a^{1}_{v_{p-1}}) = h_{v_p}$.
            \item $\phi(a^{1}_{v_p}) = h^*_q$ is a dummy house.
        \end{itemize}

        Note that such a sequence must exist and end with an agent who is assigned a dummy house, as the number of agents and houses are finite. Therefore, $p$ is a finite integer. Now, we modify the allocation to $\phi'$ as follows:

        \begin{itemize}
            \item $\phi'(a_e) = h^*_q$.
            \item $\phi'(a^{1}_{v_0}) = h_{v_0}, \phi'(a^{1}_{v_1}) = h_{v_1}, \ldots, \phi'(a^{1}_{v_{p-1}}) = h_{v_{p-1}}, \phi'(a^{1}_{v_p}) = h_{v_p}$.
            \item $\phi'(a) = \phi(a)$ for all other agents $a$.
        \end{itemize}

        We will now argue why this set of exchanges do not increase the number of envious agents.
        \begin{itemize}
            \item \textsl{Case A: $e$ is incident on $v$.} Let $e = e_1, e_2, \ldots, e_\delta$ be the edges incident on $v$. In $\phi$, $a_e$ was not envious, but agents $a_{e_2}, \ldots, a_{e_\delta}$ may or may not be envious. The agents $a^1_v, a^2_v, \ldots, a^\delta_v$ were all envious in $\phi$. In $\phi'$, $a^1_v, a^2_v, \ldots, a^\delta_v$ are not envious, agent $a_e$ becomes envious, and agents $a_{e_2}, \ldots, a_{e_\delta}$ may become envious. Therefore, the number of envious agents in $\{a_e, a_{e_2}, a_{e_3}, \ldots, a_{e_\delta}\} \cup \{a^1_v, a^2_v, \ldots, a^\delta_v\}$ do not increase. The agents $a^1_{v_1}, a^1_{v_2}, \ldots, a^1_{v_p}$ all become non-envious in $\phi'$ as they get their preferred houses. For all other agents $a^j_u$ where $u$ is a vertex in $V$, the dummy house $h^*_q$ is the only house which appears in the set of houses assigned to neighbours of $a^j_u$ in $\phi'$ which did not appear in $\phi$; hence $a^j_u$ is envious in $\phi'$  only if they were envious in $\phi$. Finally, all other agents $a_{e'}$ where $e'$ is an edge in $E$, are envious in $\phi'$ if and only if they were envious in $\phi$ as only $h_v$ (a house not preferred by $a_{e'}$) gets added to the set of houses assigned to their neighbours. 
            \item \textsl{Case B: $e$ is not incident on $v$.} Let $e_1, e_2, \ldots, e_\delta$ be the edges incident on $v$. In $\phi$, the agents $a_{e_1}, a_{e_2}, \ldots, a_{e_\delta}$ may or may not be envious; but the agents $a^1_v, a^2_v, \ldots, a^\delta_v$ were all envious in $\phi$. In $\phi'$, $a^1_v, a^2_v, \ldots, a^\delta_v$ are not envious, while agents $a_{e_1}, a_{e_2}, \ldots, a_{e_\delta}$ may become envious. Therefore, the number of envious agents in $\{a_{e_1}, a_{e_2}, a_{e_3}, \ldots, a_{e_\delta}\} \cup \{a^1_v, a^2_v, \ldots, a^\delta_v\}$ do not increase. The agents $a^1_{v_1}, a^1_{v_2}, \ldots, a^1_{v_p}$ all become non-envious in $\phi'$ as they get their preferred houses. If the agent $a_e$ is not envious in $\phi$, then it cannot become envious in $\phi'$ as only $h_{v}$ (a house not preferred by $a_e$) gets added to the set of houses assigned to its neighbours. For all other agents $a^j_u$ where $u$ is a vertex in $V$, the dummy house $h^*_q$ is the only house which appears in the set of houses assigned to neighbours of $a^j_u$ in $\phi'$ which did not appear in $\phi$; hence $a^j_u$ is envious in $\phi'$ only if they were envious in $\phi$. Finally, all other agents $a_{e'}$ where $e'$ is an edge in $E$, are envious in $\phi'$ if and only if they were envious in $\phi$ as only $h_v$ (a house not preferred by $a_{e'}$) gets added to the set of houses assigned to their neighbours. 
        \end{itemize}
    \end{itemize}

    These exchanges do not increase the total envy, but they allocate at least one house $h_v$ to $a^1_v$, which was previously allocated to some other agent. Repeating this for at most $|\HH|$ times would give us a nice allocation, with envy no more than that in $\phi$.

    Hence, we can safely assume that $\phi$ is a nice allocation with envy at most $k\delta - \binom{k}{2}$. $\phi$ allocates dummy houses to $a_e$ for all $e \in E$ and to $a^2_v, a^3_v, \ldots, a^\delta_v$ for all $v \in V$. Moreover, it is safe to assume that all dummy houses are assigned to some agent, otherwise we can assign a dummy house to any $a^1_v$ without increasing total envy.

    We now look into agents which are not assigned dummy vertices by $\phi$. There are $|\AA| - (\delta N + M - k) = (\delta N + M) - (\delta N + M -k) = k$ many such agents, all of the form $a^1_v$, $v \in V$. Let $a^{1}_{v_1}, a^{1}_{v_2}, \ldots, a^{1}_{v_k}$ be the agents which are not assigned dummy houses. For any $e \in E$, $e$ is incident on least one vertex in $v_1, v_2, \ldots, v_k$, if and only if the agent $a_e$ is envious. Therefore, the number of envious agents is precisely the number of edges which are incident on the vertices $v_1, v_2, \ldots, v_k$. Let $\gamma$ be the number of edges with both endpoints in $\{v_1, v_2, \ldots, v_k\}$. The number of edges incident on $\{v_1, v_2, \ldots, v_k\}$ is therefore equal to $k\delta - \gamma$. By assumption, the number of envious agents is at most $k\delta - \binom{k}{2}$. Thus $\gamma \ge \binom{k}{2}$, implying $v_1, v_2, \ldots, v_k$ is a clique.

    This completes the proof of correctness of the reduction.
\end{proof}

\section{Proof of Theorem~\ref{thm:3reg-hard}}
\begin{proof}
    We provide a reduction from the $1/2$-Vertex Separator problem~\cite{muller1991alpha} on 3-regular graphs to  \ohaa. The problem takes input a graph $G(V,E)$ and an integer $k$ and asks if there is a subset $S$ of $V$ that $|S| \le k$ such that $V \setminus S$ can be partitioned into two equal-sized subsets $S_1$, $S_2$ such that no edge connects a vertex in $S_1$ with a vertex in $S_2$. This problem is shown to be \NP-complete even for 3-regular graphs by M{\"u}ller and Wagner~\cite{muller1991alpha}.

    Let $(G(V,E),k)$ be an instance of $1/2$-Vertex Separator problem such that $G(V,E)$ is 3-regular. Let $V = \{a_1, a_2, \ldots, a_n\}$. We create an instance $(\AA, \HH, \GG_\AA(\AA, E_\AA), (\PP_a)_{a \in \AA}, 2\left\lfloor \frac k 2 \right\rfloor)$ of \ohaa as follows:

    \begin{itemize}
        \item $\AA = V = \{a_1, a_2, \ldots, a_n\}$.
        \item $\HH = \{h_1, h_2, \ldots, h_n\}$.
        \item $E_\AA = E$.
        \item Let $\HH^* = \{h_1, h_2, \ldots, h_t\}$, where $t = \frac{n}{2} - \left\lfloor \frac{k}{2} \right\rfloor$. $\PP_a = \HH^*$ for all $a \in \AA$. 
    \end{itemize}

    Intuitively the underlying agent graph is the same as $G$, and every agent prefers the same $t = \frac{n}{2} - \left\lfloor \frac{k}{2} \right\rfloor$ houses. Note that $t$ is an integer since the number of vertices of $G$, $n$, is even as $G$ is 3-regular. This reduction can be computed in polynomial time. We now show the correctness of the reduction. 
    
    \begin{claim}\label{clm:half-vs-equiv}
        $(G, k)$ and $(G, 2 \left\lfloor \frac{k}{2} \right\rfloor)$ are equivalent instances of the $1/2$-Vertex Separator problem, if $G$ is 3-regular.
    \end{claim}
    \begin{proof}
        This holds trivially when $k$ is even. Further, if $(G, 2 \left\lfloor \frac k 2 \right\rfloor)$ is a Yes-instance, then so is $(G, k)$ as $k \ge 2 \left\lfloor \frac k 2 \right\rfloor$.
        
        Now, say $k$ is odd. and let $(G, k)$ be a Yes-instance. Let $k = 2\lambda + 1$ for $\lambda \in \mathbb{Z}$. Note that since $G$ is 3-regular, $n = |V|$ must be even. Let $S$ be any separator that separates $V\setminus S$ into equal sized, disjoint, and non-adjacent parts $S_1$, $S_2$ and satisfying $|S| \le k$. Therefore, $|S_1| = |S_2|$. Hence $|V| - |S| = |S_1| + |S_2|$, or $|S| = n - 2|S_1|$, an even integer: say $|S| = 2 \beta$, for some integer $\beta$. Therefore, $2\beta \le 2\lambda + 1 \implies \beta \le \lambda + 0.5$. Since $\beta, \lambda$ are integers, we must have $\beta \le \lambda$; implying $|S| = 2\beta \le 2\lambda = 2 \left\lfloor \frac k 2 \right\rfloor$. Hence $(G, 2 \left\lfloor \frac k 2 \right\rfloor)$ is also a Yes-instance.  
    \end{proof} 

    \begin{claim}\label{clm:half-exact}
        If $(G(V,E), 2 \left\lfloor \frac{k}{2} \right\rfloor)$ is a Yes-instance of the $1/2$-Vertex Separator problem, where $G$ is 3-regular, there exists $S \subseteq V$ of size $|S| = 2 \left\lfloor \frac k 2 \right\rfloor$ that separates $V \setminus S$ into equal sized, disjoint, and non-adjacent parts.
    \end{claim}
    \begin{proof} 
        Let $S'$ be any separator that separates $V \setminus S$ into two equal sized, disjoint, and non-adjacent parts $S_1', S_2'$, such that $|S'| \le 2 \left\lfloor \frac k 2 \right\rfloor$. Such an $S'$ exists as $(G(V,E), 2 \left\lfloor \frac k 2 \right\rfloor)$ is a Yes-instance. 

        Let $S_1' = \{x_1, x_2, \ldots, x_s\}$ and $S_2' = \{y_1, y_2, \ldots, y_s\}$ for some $s$, where $x_i$ is not adjacent to $y_j$ for all $i,j \in [s]$. Therefore, $|S'| = n - 2s$, an even number. This gives us

        $$n - 2s \le 2 \left\lfloor \frac k 2 \right\rfloor \implies s - \frac n 2 + \left\lfloor \frac k 2 \right\rfloor \ge 0 $$

        Define $\gamma = s - \frac n 2 + \left\lfloor \frac k 2 \right\rfloor$, note that $0 \le \gamma \le s$. We now construct our target separator $S$.

        \begin{itemize}
            \item define $S_1 = \{x_{\gamma + 1}, x_{\gamma + 2}, \ldots, x_s\}$
            \item define $S_2 = \{y_{\gamma + 1}, y_{\gamma + 2}, \ldots, y_s\}$
            \item define $S = S' \cup \{x_1, x_2, \ldots, x_{\gamma}\} \cup \{y_1, y_2, \ldots, y_\gamma\}$
        \end{itemize}
         
        We have $S \cup S_1 \cup S_2 = V$, and $S$ separators $V \setminus S$ into equal sized, disjoint, and non-adjacent subsets $S_1$ and $S_2$. Moreover, 
        
        $$|S| = |S'| + 2\gamma = (n - 2s) + 2 \left(s - \frac n 2 + \left\lfloor \frac k 2 \right\rfloor\right) = 2 \left\lfloor \frac k 2 \right\rfloor$$

        This completes the proof of Claim~\ref{clm:half-exact}.
    \end{proof}

    \textbf{Reverse direction.}
    We first show if the \ohaa instance $(\AA, \HH, \GG_\AA(\AA, E_\AA), (\PP_a)_{a \in \AA}, 2 \left\lfloor \frac k 2 \right\rfloor)$ is a Yes-instance, then so is the $1/2$-Vertex Separator instance of $(G,k)$. Let $\phi$ be an allocation with at most $2\left\lfloor \frac k 2 \right\rfloor$ envious agents. Let $X \subseteq \AA$ defined by $X = \{a \in \AA \mid \phi(a) \in \HH^*\}$. Therefore, $|X| = t = \frac{n}{2} - \left\lfloor \frac k 2 \right\rfloor$. Let $J \subseteq \AA$ be the set of envious agents. By assumption, $|J| \le 2 \left\lfloor \frac k 2 \right\rfloor$. Let $Y = \AA \setminus (X \cup J)$, be the set of non-envious agents outside $X$.
    
    Observe that no edge in $E_\AA$ can connect an agent in $X$ to an agent in $Y$; otherwise the corresponding agent in $Y$ would be envious of the corresponding agent in $X$. Note that $|Y \cup J| = \frac n 2 + \left\lfloor \frac k 2 \right\rfloor \ge 2 \left\lfloor \frac k 2 \right\rfloor$. Let $J'$ be an arbitrary subset of $Y \cup J$ and a superset of $J$ satisfying $|J'| = 2 \left\lfloor \frac k 2 \right\rfloor$. Then $J'$ is separates $V = \AA$ into $X$ and $(Y \setminus J')$ satisfying $|X| = |Y \setminus J'| = \frac n 2 - \left\lfloor \frac k 2 \right\rfloor$. Therefore, $(G, 2\left\lfloor \frac k 2 \right\rfloor)$ is a Yes-instance of the $1/2$-Vertex Separator problem, and so is $(G,k)$ (by Claim~\ref{clm:half-vs-equiv}).
    
    \textbf{Forward direction.}    
    On the other hand, if $(G, k)$ is a Yes-instance for the $1/2$-Vertex Separator problem, then so is $(G, 2 \left\lfloor \frac k 2 \right\rfloor)$ (by Claim~\ref{clm:half-vs-equiv}). By Claim~\ref{clm:half-exact}, there exists a partition $S, X, Y$ of $\AA$ such that $|S| = 2 \left\lfloor \frac k 2 \right\rfloor$, $|X| = |Y| = \frac n 2 - \left\lfloor \frac k 2 \right\rfloor$, and there is no edge between $X$ and $Y$ in $\GG_\AA$. We create an allotment $\phi$ defined as follows:
    \begin{itemize}
        \item Assign all houses in $\HH^*$ to agents in $X$ (arbitrarily).
        \item Assign all other houses to all other agents (arbitrarily).
    \end{itemize}
    The only envious agents can be the ones in $S$. Hence there are at most $2 \left\lfloor \frac k 2 \right\rfloor$ many envious agents. Therefore, the \ohaa instance $(\AA, \HH, \GG_\AA(\AA, E_\AA), (\PP_a)_{a \in \AA}, 2 \left\lfloor \frac k 2 \right\rfloor)$ is a Yes-instance.

    This completes the correctness of the reduction.
\end{proof}

\section{Proof of Theorem~\ref{thm:vc-bip}}

\begin{proof}
    We reduce from the \NP-hard problem of \CLIQUE. Let $(G(V,E), k)$ be an instance of \CLIQUE, where the problem asks to decide if $G$ has a clique of size $k$. Let $|V| = N$ and $|E| = M$. Moreover, we can assume $N \ge (2M)^{\lceil 1 / \varepsilon \rceil}$; indeed if that is not the case, we can add (in polynomial time), $(2M) ^ {\lceil 1 / \varepsilon \rceil}$ many isolated vertices to $G$ without changing the size of the maximum clique.
    
    We reduce $(G(V,E), k)$ to an instance of \ohaa as follows. 

    \begin{itemize}
        \item Let $\AA = \{a_v \mid v \in V\} \cup \{a^1_e \mid e \in E\} \cup \{a^2_e \mid e \in E\}$. That is, we have an agent $a_v$ for every vertex $v$ of $G$, and two agents $a^1_e$, and $a^2_e$ for each edge $e$ of $G$. Therefore, $n = |\AA| = N + 2M$. 
        \item Let $\HH = \{h_e \mid e \in E\} \cup \{h^*_j \mid j \in [N+2M-\binom{k}{2}]\}$. Therefore, $|\HH| = M + N + 2M - \binom{k}{2} = |\AA| + (M-\binom{k}{2})$. 
        \item We are now going to set the houses $h_e$ to be preferred by agents $a^1_e$ and $a^2_e$, and also by $a_u$ and $a_v$ where $e = \{u,v\}$. This is effectively done by setting $\PP_{a^1_e} = \PP_{a^2_e} = \{h_e\}$, and $\PP_{a_v} = \{h_{e'} \mid e' \in E \text{ and } v \in e'\}$. The houses $\{h^*_j \mid j \in [N+2M - \binom{k}{2}]\}$ are preferred by no agents, i.e. these are dummy houses.
        \item We now define the underlying graph $\GG_\AA(\AA, E_\AA)$. Let $\AA_V = \{a_v \mid v \in V \}$ and let $\AA_E = \{a^1_e \mid e \in E\} \cup \{a^2_e \mid e \in E\}$. Then $\GG_\AA$ defines the complete bipartite graph between $\AA_V$ and $\AA_E$. That is, $E_\AA = \{\{a,a'\} \mid a \in \AA_V, a' \in \AA_E\}$. 
        \item The target number of envious agents is $k$.
    \end{itemize}
    Notice that the part $\AA_E$ of the complete bipartite graph $\GG_\AA$ satisfies $|\AA_E| = 2M \le N^\varepsilon = |\AA_V|^\varepsilon \le |\AA|^\varepsilon$. This reduction can be done in polynomial time. We now show the correctness of the reduction.

    \textbf{Forward direction.} 
    Let $(G(V, E), k)$ be a Yes-instance of \CLIQUE, i.e., there is clique $S \subseteq V$ of size $k$. Let $E_S$ be the edges of the clique $S$, $|E_S| = \binom{k}{2}$. Consider the allocation $\phi$ defined as:
    \begin{itemize}
        \item A dummy house is assigned to $a_v$ for all $v \in V$. This uses $N$ dummy houses.
        \item A dummy house is assigned to $a^2_e$ for all $e \in E$. This uses $M$ dummy houses.
        \item A dummy house is assigned to $a^1_e$ for every $e \in E \setminus E_S$. This uses $(M - \binom{k}{2})$ dummy houses. 
        \item For every $e \in E_S$,the house $h_e$ is assigned to $a^1_e$. Hence the houses $\{h_e \mid e \in E \setminus E_S\}$ are all unallocated.
    \end{itemize}

    Note that for every $e \in E$, the agents $a^1_e$ and $a^2_e$ are all non-envious, because their neighbours in $\GG_\AA$, i.e. the agents $\{a_v \mid v \in V\}$ are all assigned dummy houses. Moreover, for all $v \in V$, such that $v \notin S$, no edge in $E_S$ is incident to $v$; thus all houses preferred by $a_v$ are unallocated, making $a_v$ non-envious. Therefore, all envious agents are in $\{a_v \mid v \in S\}$. This makes the number of envious agents in $\phi$ to be at most $|S| = k$; the \ohaa instance is a Yes-instance.

    \textbf{Reverse direction.}
    Let $\phi$ be some allocation with at most $k$ envious agents. We say that $\phi$ is nice if for all $e \in E$, $h_e$ is either unassigned or is assigned to $a^1_e$. If $\phi$ is not nice to begin with, then $\phi$ can be converted to a nice allocation by a sequence of exchanges, none of which increases total envy. 
    
    Let $h_e$ be some house which is neither unallocated, nor allocated to $a^1_e$. 
    
    \begin{itemize}
        \item \textbf{Case I: $h_e$ is assigned to $a^2_e$.} We swap the houses allocated to $a^1_e$ and $a^2_e$. This does not change the number of envious agents, due to the symmetry of $\GG_\AA$.
        \item \textbf{Case II: $h_e$ is assigned to $a^j_{e'}$ for some $j \in \{1,2\}$, $e' \ne e$.} We swap the houses allocated to $a^1_e$ and $a^j_{e'}$. Again, this does not increase the number of envious agents. This is because
        \begin{enumerate}
            \item $a^1_e$ is not envious in both allocations.
            \item $a^j_{e'}$ did not have their preferred house before the swap. Since the neighbourhood of $a^j_{e'}$ is unaffected by the swap, $a^j_{e'}$ cannot become envious after the swap, if $a^j_{e'}$ was not envious before the swap.
            \item For all other agents, they are envious after the swap if and only if they were envious before the swap. This is because the set of houses allocated to their neighbours remain unchanged.
        \end{enumerate}
        \item \textbf{Case III: $h_e$ is assigned to $a_v$ for some $v \in V$.} Agents $a^1_e, a^2_e$ must be envious of $a_v$ in $\phi$. We now look at the sequence of agents $a^{1}_{e_0} = a^1_e, a^{1}_{e_1}, \ldots, a^{1}_{e_{p-1}}, a^{1}_{e_p}$ such that 
        \begin{itemize}
            \item $\phi(a_v) = h_{e_0}$.
            \item $\phi(a^{1}_{e_0}) = h_{e_1}, \phi(a^{1}_{e_1}) = h_{e_2}, \ldots, \phi(a^{1}_{e_{p-1}}) = h_{e_p}$.
            \item $\phi(a^{1}_{e_p}) = h^*_q$ is a dummy house.
        \end{itemize}

        Note that such a sequence must exist and end with an agent who is assigned a dummy house, as the number of agents and houses are finite. Therefore, $p$ is a finite integer. Now, we modify the allocation to $\phi'$ as follows:

        \begin{itemize}
            \item $\phi'(a_v) = h^*_q$.
            \item $\phi'(a^{1}_{e_0}) = h_{e_0}, \phi'(a^{1}_{e_1}) = h_{e_1}, \ldots, \phi'(a^{1}_{e_{p-1}}) = h_{e_{p-1}}, \phi'(a^{1}_{e_p}) = h_{e_p}$.
            \item $\phi'(a) = \phi(a)$ for all other agents $a$.
        \end{itemize}

        We will now argue why this set of exchanges do not increase the number of envious agents.
        \begin{itemize}
            \item \textsl{Case A: $e$ is incident on $v$.} Let $e = \{u,v\}$. In $\phi$, $a_v$ was not envious, but $a_u$ may or may not be envious. The agents $a^1_e$, and $a^2_e$ were both envious. In $\phi'$, $a^1_e$ and $a^2_e$ are not envious, agent $a_v$ becomes envious, and $a_u$ may become envious. Therefore, the number of envious agents in $\{a^1_e, a^2_e, a_u, a_v\}$ do not increase. The agents $a^1_{e_1}, a^1_{e_2}, \ldots, a^1_{e_p}$ all become non-envious in $\phi'$ as they get their preferred houses. For all other agents $a^j_{e'}$ where $e'$ is an edge in $E$, the dummy house $h^*_q$ is the only house which appears in the set of houses assigned to neighbours of $a^j_{e'}$ in $\phi'$ which did not appear in $\phi$; hence $a^j_{e'}$ is envious in $\phi'$ only if they were envious in $\phi$. Finally, all other agents $a_{v'}$ where $v'$ is a vertex in $V$, are envious in $\phi'$ if and only if they are envious in $\phi$ as only the house $h_v$ (a house not preferred by $a_{v'}$) gets added to the set of houses assigned to the neighbours of $a_{v'}$. 
            \item \textsl{Case B: $e$ is not incident on $v$.} Let $e=\{u_1, u_2\}$. In $\phi$, the agents $a_{u_1}$, $a_{u_2}$ may or may not be envious; but the agents $a^1_e$ and $a^2_e$ were both envious. In $\phi'$, $a^1_e$ and $a^2_e$ are not envious, while agents $a_{u_1}$ and $a_{u_2}$ may become envious. Therefore, the number of envious agents in $\{a_{u_1}, a_{u_2}, a^1_e, a^2_e\}$ do not increase. The agents $a^1_{e_1}, a^1_{e_2}, \ldots, a^1_{e_p}$ all become non-envious in $\phi'$ as they get their preferred houses. If the agent $a_v$ is not envious in $\phi$, then it cannot become envious in $\phi'$ as only $h_{e}$ (a house not preferred by $a_v$) gets added to the set of houses assigned to its neighbours. For all other agents $a^j_{e'}$ where $e'$ is an edge in $E$, the dummy house $h^*_q$ is the only house which appears in the set of houses assigned to neighbours of $a^j_{e'}$ in $\phi'$ which did not appear in $\phi$; hence $a^j_{e'}$ is envious in $\phi'$ only if they were envious in $\phi$. Finally, all other agents $a_{v'}$ where $v'$ is a vertex in $V$, are envious in $\phi'$ if and only if they were envious in $\phi$ as only the house $h_v$ (a house not preferred by $a_{v'}$) gets added to the set of houses assigned to the neighbours of $a_{v'}$. 
        \end{itemize}
    \end{itemize}
    These exchanges do not increase the total envy, but they allocate at least one house $h_e$ to $a^1_e$, which was previously allocated to some other agent. Repeating this for at most $|\HH|$ times would give us a nice allocation, with envy no more than that in $\phi$.

    Hence, we can safely assume that $\phi$ is a nice allocation with envy at most $k$. $\phi$ allocates dummy houses to $a_v$ for all $v \in V$ and to $a^2_e$ for all $e \in E$. Moreover, it is safe to assume that all dummy houses are assigned to some agent, otherwise we can assign a dummy house to any $a^1_e$ without increasing total envy.

    We now look into agents which are not assigned dummy vertices by $\phi$. There are $N + 2M - (N + 2M - \binom{k}{2}) = \binom{k}{2}$ many such agents, each of the form $a^1_e$ for some $e \in E$. Let $\gamma = \binom{k}{2}$. Let $a^{1}_{e_1}, a^{1}_{e_2}, \ldots, a^{1}_{e_\gamma}$ be the agents which are not assigned dummy houses. For any $v \in V$, $v$ is an endpoint of at least one edge in $e_1, e_2, \ldots, e_\gamma$, if and only if the agent $a_v$ is envious. Therefore, the number of envious agents is precisely the number of vertices on which the edges $e_1, e_2, \ldots, e_\gamma$ are incident on. The number of vertices that $\binom{k}{2}$ distinct edges are incident on, is at least $k$, with equality holding if and only if the edges induce a clique of size $k$. Further, the number of envious agents, i.e. the number of vertices on which the edges $e_1, e_2, \ldots, e_\gamma$ are incident on, is at most $k$ by assumption. Therefore indeed, $\{e_1, e_2, \ldots, e_\gamma\}$ is a set of $\binom{k}{2}$ distinct edges, which induce a clique of size $k$. 

    This completes the proof of correctness of the reduction.
\end{proof}

\section{NP-Hardness On Split Graphs with Bounded Vertex Cover}\label{App:vc-split}

\begin{theorem}\label{thm:vc-split}
    Let $\varepsilon \in (0, 1)$ be any constant. Consider an input $(\AA, \HH, \GG_\AA(\AA, E_\AA), (\PP_a)_{a \in \AA})$ with $m$ houses and $n$ agents. \ohaa is \NP-hard even when $\GG_\AA$ is a split graph where the maximum clique is of size at most $n^\varepsilon$, and each house is preferred by at most three agents.
\end{theorem}

\begin{proof}
    We provide a reduction from the \CLIQUE~problem to \ohaa. Let $(G(V,E), k)$ be an instance of \CLIQUE, which asks to decide whether $G$ has a clique of size $k$. 

    Let $|V| = N, |E| = M \ge 1$. Define $t = N ^{\lceil 1 / \varepsilon \rceil}$. We now construct an instance of \ohaa as follows:
    
    \begin{itemize}
        \item Let $\AA = \{a_v \mid v \in V\} \cup \{a^j_e \mid e \in E, j \in [t]\}$. That is, we have an agent $a_v$ for every vertex $v$ of $G$, and $t$ agents $a^1_e, a^2_e, \ldots, a^t_e$ for each edge $e$ of $G$. Therefore, $n = |\AA| = N + M \cdot t = N + M \cdot N ^{\lceil 1/\varepsilon \rceil}$. 
        \item Let $\HH = \{h^j_e \mid e \in E, j \in [t]\} \cup \{h^*_j \mid j \in [(M - \binom{k}{2})t + N]\}$. Therefore, $|\HH| = M \cdot t + N + (M - \binom{k}{2})t = |\AA| + (M - \binom{k}{2})t$. 
        \item We are now going to set the houses $h^j_e$ to be preferred by agents $a^j_e$, $a_{u}, a_{v}$, where $e = \{u, v\}$. This is effectively done by setting $\PP_{a^j_e} = \{h^j_e\}$, and $\PP_{a_v} = \{h^j_e \mid j \in [t], e \in E \text{ such that } v \in e\}$. The houses $\{h^*_j \mid j \in [(M - \binom{k}{2})t + N]\}$ are preferred by no agents, i.e., these are dummy houses.
        \item We now define the underlying graph $\GG_\AA(\AA, E_\AA)$. The edges in $E_\AA$ are defined as follows:
        \begin{itemize}
            \item For every $u, v \in V$, $\{a_u, a_v\} \in E_\AA$.
            \item For every $v \in V, e \in E, j \in [t]$, $\{a_{v}, a^j_e\} \in E_\AA$.
        \end{itemize}
        \item The target number of envious agents is $k$.
    \end{itemize}

    This means that $C = \{a_v \mid v \in V\}$ forms a clique and $I = \{a^j_e \mid e \in E, j \in [t]\}$ forms an independent set with $\AA = C \cup I$. Moreover, $|C| = N < (N + M\cdot N^{\lceil 1 / \varepsilon \rceil})^\varepsilon = n^\varepsilon$. The output of the reduction is the \ohaa instance $(\AA, \HH, \GG_\AA, (\PP_a)_{a \in \AA}, k)$. This reduction can be done in polynomial time as $\varepsilon$ is a constant. We now show the correctness of this reduction.

    \textbf{Forward direction.}
    Let $(G(V, E), k)$ be a Yes-instance of \CLIQUE, i.e., there is a clique $S \subseteq V$ of size $k$. Let $E_S$ be the edges of the clique $S$. We have $|E_S| = \binom{k}{2}$. Consider the allocation $\phi$ defined as:
    \begin{itemize}
        \item A dummy house is assigned to $a_v$ for all $v \in V$. This uses $N$ dummy houses.
        \item A dummy house is assigned to $a^j_e$ for every $j \in [t], e \in E \setminus E_S$. This uses $(M - \binom{k}{2})t$ dummy houses. 
        \item For every $j \in [t], e \in E_S$, the house $h^j_e$ is assigned to $a^j_e$. Hence the houses $\{h^j_e \mid j \in [t], e \in E \setminus E_S\}$ are all unallocated.
    \end{itemize}

    Note that for every $e \in E$, the agents $a^j_e$ are all non-envious, because their neighbours in $\GG_\AA$, i.e., the agents $\{a_v \mid v \in V\}$ are all assigned dummy houses. Moreover, for all $v \in V$, such that $v \notin S$, no edge in $E_S$ is incident to $v$; thus, all houses preferred by $a_v$ are unallocated, making $a_v$ non-envious. Therefore, all envious agents are in $\{a_v \mid v \in S\}$. This makes the number of envious agents in $\phi$ to be at most $|S| = k$; the \ohaa instance is a Yes-instance. 

    \textbf{Reverse direction.}
    Let $\phi$ be some allocation with at most $k$ envious agents. We say that $\phi$ is nice if for all $j \in [t], e \in E$, $h^j_e$ is either unassigned or is assigned to $a^j_e$. If $\phi$ is not nice to begin with, then $\phi$ can be converted to a nice allocation by a sequence of exchanges, none of which increases total envy. 
    
    Let $h^j_e$ be some house which is neither unallocated, nor allocated to $a^j_e$. 
    
    \begin{itemize}
        \item \textbf{Case I: $h^j_e$ is allocated to $a^{j'}_{e'}$ for some $j' \ne j$ or $e' \ne e$.} Then we swap houses allocated to $a^j_e$ and $a^{j'}_{e'}$. After this swap, $a^j_e$ is non-envious by getting a preferred house. The agent $a^{j'}_{e'}$ cannot become envious after the swap if $a^{j'}_{e'}$ was not envious before the swap; this is because the houses allocated to the neighbours of $a^{j'}_{e'}$ remain unchanged and $a^{j'}_{e'}$ did not have a preferred house before the swap. For all other agents, they remain envious after the swap if and only if they were envious before the swap (as the set of houses allocated to their neighbours is unchanged).
        \item \textbf{Case II: $h^j_e$ is allocated to $a_v$ for some $v \in V$.} $a^j_e$ must be envious of $a_v$ in $\phi$. We now look at the sequence of agents $a^j_e = a^{j_0}_{e_0}, a^{j_1}_{e_1}, \ldots, a^{j_{p-1}}_{e_{p-1}}, a^{j_p}_{e_p}$ such that 
        \begin{itemize}
            \item $\phi(a_v) = h^{j_0}_{e_0}$.
            \item $\phi(a^{j_0}_{e_0}) = h^{j_1}_{e_1}, \phi(a^{j_1}_{e_1}) = h^{j_2}_{e_2}, \ldots, \phi(a^{j_{p-1}}_{e_{p-1}}) = h^{j_p}_{e_p}$.
            \item $\phi(a^{j_p}_{e_p}) = h^*_q$ is a dummy house.
        \end{itemize}
        Note that such a sequence must exist and end with an agent who is assigned a dummy house, as the number of agents and houses are finite. Therefore, $p$ is a finite integer. Now, we modify the allocation to $\phi'$ as follows:

        \begin{itemize}
            \item $\phi'(a_v) = h^*_q$.
            \item $\phi'(a^{j_0}_{e_0}) = h^{j_0}_{e_0}, \phi'(a^{j_1}_{e_1}) = h^{j_1}_{e_1}, \ldots, \phi'(a^{j_{p-1}}_{e_{p-1}}) = h^{j_{p-1}}_{e_{p-1}}, \phi'(a^{j_p}_{e_p}) = h^{j_p}_{e_p}$.
            \item $\phi'(a) = \phi(a)$ for all other agents $a$.
        \end{itemize}

        In $\phi'$, $a^{j}_e$ is not envious anymore, however, $a_v$ may or may not become envious. The agents $a^{j_1}_{e_1}, \ldots, a^{j_{p-1}}_{e_{p-1}}, a^{j_p}_{e_p}$, were allocated their preferred house in $\phi'$, and hence are not envious. Moreover, for all other $a^{j'}_{e'}$, they are envious in $\phi'$ if and only if they were envious in $\phi$; this is because the dummy house $h^*_q$ is the only house that appears in the set of houses allocated to their neighbours in $\phi'$ but not in $\phi$. The same happens for $a_{v'}$, $v' \in V \setminus \{v\}$; the set of houses allocated to the neighbours of $a_{v'}$ does not change. Therefore this does not increase total envy.
    \end{itemize}

    These exchanges do not increase the total envy, but they allocate at least one house $h^j_e$ to $a^j_e$, which was previously allocated to some other agent. Repeating this for at most $|\HH|$ times would give us a nice allocation, with envy no more than that in $\phi$.

    Hence, we can safely assume that $\phi$ is a nice allocation with envy at most $k$. $\phi$ allocates dummy houses to $a_v$ for all $v \in V$. Moreover, it is safe to assume that all dummy houses are assigned to some agent, otherwise we can assign a dummy house to any $a^j_e$ without increasing total envy.

    We now look into agents which are not assigned dummy vertices by $\phi$. There are $Mt - (M - \binom{k}{2})t = \binom{k}{2}t$ many such agents. Let $\gamma = \binom{k}{2}t$. Let $a^{j_1}_{e_1}, a^{j_2}_{e_2}, \ldots, a^{j_\gamma}_{e_\gamma}$ be the agents which are not assigned dummy houses. Therefore, the number of distinct edges in $e_1, e_2, \ldots, e_\gamma$ is at least $\gamma / t = \binom{k}{2}$. Moreover, for any $v \in V$, $v$ is an endpoint of at least one edge in $e_1, e_2, \ldots, e_\gamma$, if and only if the agent $a_v$ is envious. Therefore, the number of envious agents is precisely the number of vertices on which the edges $e_1, e_2, \ldots, e_\gamma$ are incident on. The number of vertices that $\binom{k}{2}$ distinct edges are incident on, is at least $k$, with equality holding if and only if the edges induce a clique of size $k$. However, the number of envious agents, i.e., the number of vertices on which the edges $e_1, e_2, \ldots, e_\gamma$ are incident on, is at most $k$ by assumption. Therefore, $\{e_1, e_2, \ldots, e_\gamma\}$ is a set of $\binom{k}{2}$ distinct edges, which induce a clique of size $k$. 

    This completes the proof of correctness of the reduction.
\end{proof}}
\end{document}